\documentclass[journal]{IEEEtran}
\usepackage{hyperref}
\usepackage{amsmath}
\usepackage{amsfonts}
\usepackage{amssymb}
\usepackage{amscd}
\usepackage[dvips]{graphicx}
\usepackage{tikz}
\def\BibTeX{{\rm B\kern-.05em{\sc i\kern-.025em b}\kern-.08em
		T\kern-.1667em\lower.7ex\hbox{E}\kern-.125emX}}

\usepackage{amsthm}
\usepackage{newlfont}
\usepackage{color}
\usepackage{setspace}
\usepackage{algorithm}
\usepackage{algorithmic}
\usepackage{cite}
\usepackage{epstopdf}

\newtheorem{lem}{Lemma}
\newtheorem{prop}{Proposition}

\newcommand{\beq}{\begin{equation}}
	\newcommand{\eeq}{\end{equation}}
\newcommand{\beqa}{\begin{eqnarray}}
	\newcommand{\eeqa}{\end{eqnarray}}

\begin{document}

\title{Joint Collaboration and Compression Design for Distributed Sequential Estimation in a Wireless Sensor Network}
\author{Xiancheng~Cheng,
        Prashant Khanduri,
        Boxiao Chen,
        Pramod K.Varshney, \IEEEmembership{Life~Fellow,~IEEE}
}


\maketitle

\begin{abstract}
In this work, we propose a joint collaboration-compression framework for sequential estimation of a random vector parameter in a resource constrained wireless sensor network (WSN). Specifically, we propose a framework where the local sensors first collaborate (via a collaboration matrix) with each other. Then a subset of sensors selected to communicate with the FC linearly compress their observations before transmission. 
We design near-optimal collaboration and linear compression strategies under power constraints via alternating minimization of the sequential minimum mean square error. 
We show that the objective function for collaboration design can be non-convex depending on the network topology. 
We reformulate and solve the collaboration design problem using quadratically constrained quadratic program (QCQP). Moreover, the compression design problem is also formulated as a QCQP. We propose two versions of compression design, one centralized where the compression strategies are derived at the FC and the other decentralized, where the local sensors compute their individual compression matrices independently.
It is noted that the design of decentralized compression strategy is a non-convex problem. We obtain a near-optimal solution by using the bisection method.
In contrast to the one-shot estimator, our proposed algorithm is capable of handling dynamic system parameters such as channel gains and energy constraints. Importantly, we show that the proposed methods can also be used for estimating time-varying random vector parameters. Finally, numerical results are provided to demonstrate the effectiveness of the proposed framework. 
\end{abstract}


\IEEEpeerreviewmaketitle

\section{Introduction}
\label{Sec: Introduction}
\IEEEPARstart{W}{ireless} sensor networks (WSNs) are widely used for inference in a range of applications including environmental monitoring, military surveillance and health \cite{yick2008wireless, khalfallah2013new, hefeeda2007wireless, liu2009sensor, cheng2012survey}. In a WSN, a set of distributed sensors referred to as local sensors collaborate to infer a phenomenon of interest with the help of a fusion center (FC). Due to their flexibility and fault tolerance, WSNs are generally used for target detection \cite{ramson2017applications, ray2008distributed, zhang2019distributed}, target tracking \cite{cao2016sensor, shen2013sensor}, and distributed estimation \cite{akhtar2017distributed,xiao2008linear} problems. In this work, we consider sequential estimation of a random parameter vector via a resource constrained WSN. 

The sensors deployed in a WSN are usually power constrained and computationally limited devices. The local sensors obtain observations about a phenomenon of interest and forward their observations to the FC after some local processing. Moreover, the channels over which the sensors communicate with each other and with the FC are noisy and bandwidth limited. Therefore, it is necessary for the local sensors to deploy their resources in a carefully designed manner so as to meet the stringent power and bandwidth constraints. 
To this end, in this work we consider a collaboration-compression framework where the sensors are allowed to first collaborate with each other and then only a selected subset of sensors communicates with the FC. 
The communication strategies must also be designed so as to satisfy the power and bandwidth constraints of the WSN. 


Inter-sensor communication referred to as Sensor Collaboration first proposed in \cite{kar2013linear}, and then used in different frameworks \cite{liu2015sparsity, liu2014optimal, liu2016optimized, zhang2018optimal,Khanduri2016_SPL} is used to reduce the burden of communication between the local sensors and the FC. In sensor collaboration, all the sensors share their observations with the other sensor nodes defined by the collaboration matrix, after which a subset of sensors is selected for communicating their processed data to the FC. Usually, the communication cost between the local sensors and the FC is much higher compared to the sensor to sensor communication. 
Therefore, collaboration among sensors reduces the overall communication cost in a WSN. In addition, collaboration can also smooth out the observation noise, thereby enhancing the quality of data sent to the FC and improving overall inference performance. This work focuses on the design of linear collaboration strategies among sensor nodes. 

In addition to linear spatial collaboration, to reduce the communication costs further, the local sensors compress their observations before transmitting them to the FC. 
Compression strategies are designed to minimize the amount of data being transmitted to the FC.
This compression can be achieved via quantization  \cite{li2007distributed, ozdemir2008channel, fang2008distributed, collaboration_quantization}, where only symbols from a finite set are transmitted to the FC. Another, popular way of achieving compression is via linear precoding where the dimensions of the observations are reduced via a compression matrix before transmission to the FC \cite{xiao2008linear,akhtar2017distributed,Fang_Li_TSP_2012, Fang_Li_TAES_2014, Bianchi_TSP_2011, Khanduri2019_TSP, Khanduri2019_SPAWC,behbahani2012linear, shirazinia2016massive, fang2013joint}. 
In this work, we focus on the design of such linear compression strategies for the sequential estimation problem. 
In summary, the design of both collaboration and linear compression strategies is considered for distributed vector parameter estimation by minimizing the sequential mean square error. To the best of our knowledge, this is the first work that designs collaboration and compression strategies jointly in a resource constrained WSN to estimate the random vector parameter in an online fashion.



\subsubsection*{Related Literature}
Spatial collaboration was initially proposed for estimating a random scalar in \cite{kar2013linear}. It is shown that a sparse network can achieve performance identical to that of a fully connected network. In \cite{liu2014optimal}, the problem of distributed estimation with sensor collaboration is studied where the optimal sparse collaboration topology subject to information and energy constraints is designed. The sensor collaboration problem in \cite{liu2014optimal} is formulated as a sparsity-aware optimization problem by establishing a correspondence between the collaboration topology and the sparsity structure of the collaboration matrix.
Further, in \cite{liu2015sparsity} the problem of sensor selection is considered for a distributed estimation system with sensor collaboration. Specifically, optimal sensor collaboration and selection schemes are jointly designed through entry- and group-level sparsity of the collaboration matrix. The work empirically showed the trade-off between sensor collaboration and sensor selection.
Then, motivated by the monitoring of temporally correlated parameters such as daily temperature, precipitation, soil moisture and seismic activities \cite{vuran2004spatio, vuran2006spatio}, the work in \cite{liu2016optimized} considered the problem of sensor collaboration for the estimation of time-varying parameters. The optimal sensor collaboration strategy is designed based on the prior knowledge about parameter correlations. Recently, the work in \cite{zhang2018optimal} studied the tracking of a dynamic parameter which follows a first-order Gauss-Markov process with an energy-constrained sensor network. The optimal sensor collaboration strategy is designed in the presence of noisy sensor to sensor communication channels. In all the above mentioned works, the parameter of interest is assumed to be a scalar. In contrast this work considers the estimation of a random vector ("static" as well as "dynamic") parameter with noisy collaboration among local sensors.

Besides sensor collaboration, a number of works also focus on distributed estimation algorithms for reducing the communication cost for resource constrained WSNs via linear compression \cite{xiao2008linear, behbahani2012linear, liu2017joint, shirazinia2016massive, fang2013joint, kar2014decentralized}. 
In \cite{xiao2008linear}, the problem of distributed estimation of an unknown vector signal in resource constrained WSNs via coherent multiple access channels (MAC) is studied. 
The optimal encoder is designed under the criterion of
minimum mean square error (MMSE) where the observation and the channel fading matrices are both assumed to be fixed. A similar problem is discussed in \cite{behbahani2012linear}, in the presence of noisy FC.
Further, the FC equipped with a massive multiple-input multiple-output antenna system is considered. The amplification factor at each sensor node is optimized under the criterion of minimizing the total power consumption. In \cite{fang2013joint}, the distributed estimation for correlated sources is analyzed where the correlated data from multiple sensors are transmitted to the destination via orthogonal channels.
Based on the criterion of maximizing the mutual information between the sources and the received signals at the FC, the linear precoders at the sensors are jointly designed with the knowledge of the instantaneous channel state information (CSI). The aforementioned methods mainly consider the design of one-shot estimators. And the CSI and the observation matrix are usually assumed to be known a \textit{priori} as it is difficult to handle the time-varying scenarios even though they naturally occur in the WSNs. 
Recently in \cite{akhtar2017distributed}, the problem of sequential estimation in a dynamic setting is investigated. A fast block coordinate descent based precoder is designed under the criterion of sequential linear minimum mean square error (LMMSE).
In this work, we focus on designing not only such compression strategies but also focus on sensor collaboration. As pointed out earlier, this is the first work to consider the design of both collaboration and compression jointly in a WSN to estimate the random vector parameter in an online fashion.

Specifically, prior to compression, the sensors first collaborate with each other to compute the data to be sent to the FC. Then, a subset of sensors are selected to transmit their observations to the FC after compressing them into a low dimensional subspace. 
The goal here is to design the collaboration and compression strategies for a dynamic system, i.e., where the system parameters such as the observation matrix, channel gain and power constraints at each sensor can be time-varying. In practice, the resources at each sensor are limited, and the goal is to find the balance between collaboration and compression, thereby improving the performance of the system which makes it important to jointly design the collaboration and compression strategies.

To summarize, we consider the problem of sequential distributed parameter vector estimation in WSNs. The key contributions of the work are listed as follows.
\begin{enumerate}
    \item We construct a novel collaboration-compression framework for distributed estimation over resource constrained WSNs. Collaboration and compression are jointly carried out for reducing communication costs while providing excellent estimation performance.
    \item We develop a recursive linear minimum mean square error (R-LMMSE) estimator under the proposed framework for sequentially estimating a random vector with known mean and variance in a dynamic setting. We design collaboration and compression strategies to minimize the R-LMMSE. We also extend the proposed framework for online estimation of a time-varying random vector parameter.
    \item We propose online algorithms for jointly designing the optimal collaboration and compression strategies in both centralized and decentralized compression settings.
    Also, we show the convergence of the proposed estimator.
\end{enumerate}
The rest of this paper is organized as follows. 
In Section \ref{Sec: Problem statement}, we introduce the proposed collaboration-compression framework. The parameter estimation problem is formulated via R-LMMSE minimization. 
In Section \ref{Sec: Optimal collaboration and compression matrix design}, we jointly design collaboration and compression strategies. Both centralized and decentralized algorithms are presented in this section. For comparison with the proposed algorithms, we also provide a benchmark algorithm which assumes that all the observations are available at the FC. 
 In Section \ref{Sec: Numerical results}, we demonstrate the effectiveness of the proposed framework and algorithms through numerical experiments and compare them with the benchmark algorithm presented in Section \ref{Sec: Optimal collaboration and compression matrix design}. In Section \ref{Sec: Time-varying parameters estimation}, we study the problem of tracking time-varying random vector parameters under the proposed framework. Finally, we conclude the work in Section \ref{Sec: Conclusion}.

\section{Problem statement}
\label{Sec: Problem statement}
In this section, we describe the problem of parameter estimation based on the collaboration-compression scheme. In the proposed framework, instead of transmitting all the observations from individual sensor nodes to the FC directly, each sensor node first performs spatial collaboration via a coherent MAC \cite{zhang2018optimal,kar2013linear} after observing the parameter of interest through a linear measurement model.
Then the measurements after collaboration are linearly compressed \cite{akhtar2017distributed} and transmitted to the FC through a MAC. 
Due to the fact that each sensor is power constrained, a fraction of the energy is devoted to collaboration while the rest of it is used for communication with the FC, which leads to the optimal energy allocation problem among them. 
Our novel framework of joint collaboration-compression is depicted in Fig. \ref{fig_sim}. Also, some notations for the system model are provided in Table. \ref{table:1}.

\begin{table}[h!]
\centering
\caption{Notation for System Model}
\begin{tabular}{|c| c| c |} 
\hline
Parameters &    Symbol&  Space \\  \hline
$\mathbf{x}$   & Unknown parameter  &  $\mathbb{R}^{P\times 1}$  \\ \hline
$\mathbf{H}(k)$      & Observation matrix & $\mathbb{R}^{L\times P}$ \\ \hline
$\mathbf{A}$      &Network topology matrix & $\mathbb{R}^{M\times N}$ \\ \hline
$\mathbf{W}(k)$      & Collaboration matrix & $\mathbb{R}^{M\times N}$ \\ \hline
$\mathbf{F}(k)$      & Compression matrix & $\mathbb{R}^{M\times ML}$ \\ \hline
$\mathbf{G}(k)$      & Sensor-FC channel gain matrix & $\mathbb{R}^{S\times M}$ \\ \hline
\end{tabular}
\label{table:1}
\end{table}

\subsection{System model}
We consider the estimation of a random vector $\mathbf{x}\in \mathbb{R}^{P\times 1}$ through a wireless sensor network (WSN), 
whose mean $\mathbb{E}\{\mathbf{x}\} = \mathbf{x}_0$ and covariance $\mathbf{R}_{{x}} = \mathbb{E}\{\mathbf{x} {\mathbf{x}}^T\}$ are known. At each time instant $k$, the parameter of interest $\mathbf{x}$ is observed by $N$ distributed sensors through a linear measurement matrix. The observations obtained at the $i$th sensor are modeled as                                     
\begin{equation}
     \mathbf{y}_i(k) = \mathbf{H}_i(k)\mathbf{x}+\mathbf{v}_i(k),\ i\in [1,N]
\end{equation}
where $\mathbf{H}_i(k)\in \mathbb{R}^{L\times P}$ represents the linear observation matrix and $\mathbf{v}_i(k)\in \mathbb{R}^{L\times 1}$ is independent identically distributed (i.i.d) additive Gaussian noise with zero mean and covariance $\textbf{R}_{\textbf{v}_i} = \mathbb{E}\{\textbf{v}_i(k)\textbf{v}_i^T(k)\}$ at the $i$th sensor.


\begin{figure}[!t]
\centering
\includegraphics[width=3.3in, height=1.3in]{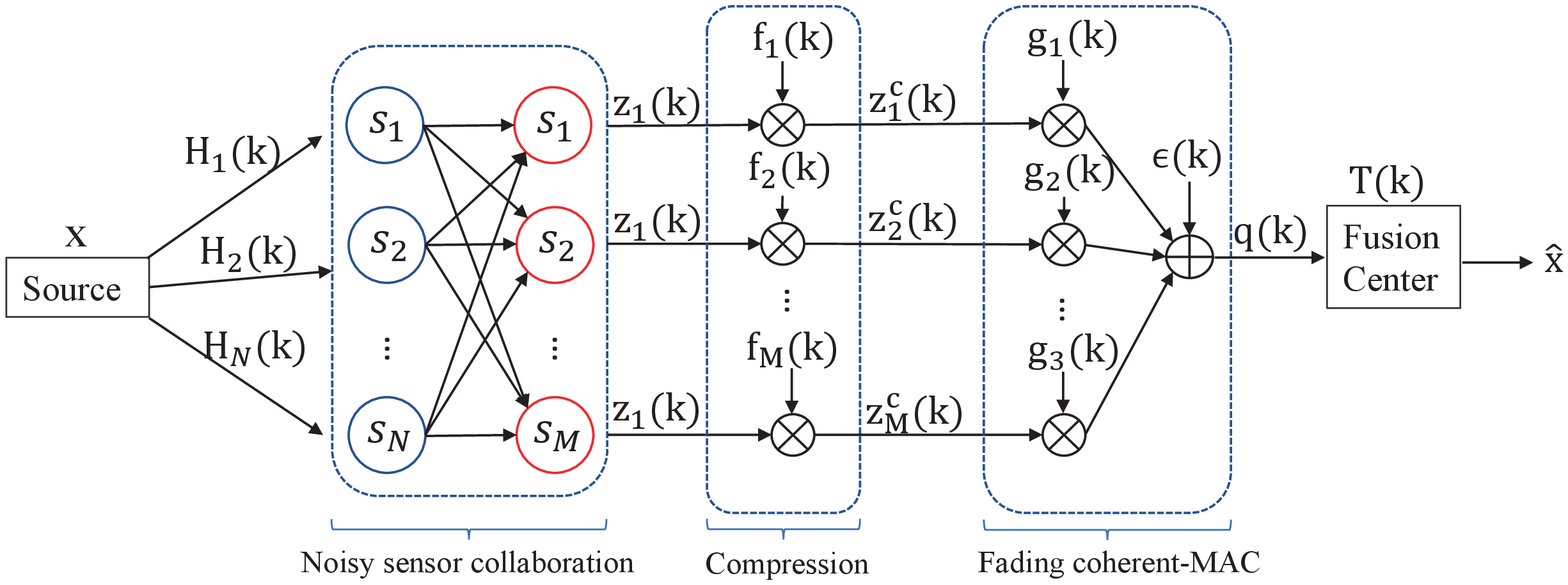}
\caption{Structure of Collaboration-compression framework}
\label{fig_sim}
\end{figure}

As discussed earlier, usually the communication cost between the sensors and the FC is much more expensive compared to the cost of inter-sensor communications. Therefore, instead of transmitting all the observations to the FC directly, a subset of ${M}$ sensors selected from ${N}$ sensors are tasked to communicate with the FC. 
Note that the choice of $M$ sensors depends on factors like proximity to the FC or the quality of channels from the local sensors to the FC. 

We assume that the network topology is fixed\footnote{In order to obtain the topology of the sensor network,
optimal sensor selection could be conducted first by minimizing the trace of the inverse of the Bayesian Fisher information matrix \cite{liu2016sensor}. 
Then the topology matrix of the sensor network $\textbf{A}$ can be further determined based on the quality of sensor observations and channel performance.} and is represented by matrix \textbf{A} with binary entries, that is to say, $A_{ij}\in \{0,1\}$ for $i\in[1,M]$ and $j\in[1,N]$ where $[1,M]$ denotes $\{1,2,\dots,M\}$. Specifically, $A_{ij}=1$ means that there is a communication link from the $j$th sensor to $i$th sensor, otherwise we have $A_{ij}=0$. 
We assume that $A_{ii}=1$ for all $i\in [1,M]$ as each sensor can collaborate with itself. If $M=N$ and $A_{ij}=0$ for all $i\neq j$, then the sensor collaboration scheme reduces to the basic amplify-and-forward (or compress and forward) transmission strategy introduced in \cite{akhtar2017distributed, xiao2008linear}. Without loss of generality, we assume that the sensors labeled $1$ to $M$ are tasked to communicate with the FC while the rest of the sensors labelled from $M+1$ to $N$ only participate in collaboration determined by the network topology.

Given the network topology, the structure of the sensor collaboration matrix $\textbf{W}(k)$ at each time instant $k$ is defined as:
\begin{equation}
    \mathbf{W}(k) \odot (\mathbf{1}_M \mathbf{1}_N^T - \mathbf{A}) = \mathbf{0}
\end{equation}
where $\mathbf{W}(k)\in \mathbb{R}^{M\times N}$ is the matrix of collaboration weights, $\odot$ denotes the Hadamard product, $\textbf{1}_M$ is the $M\times 1$ vector of all ones, and $\textbf{0}$ is the $M\times N$ matrix of all zeros. The signal at each sensor node after collaboration at time $k$ is modeled as:
\begin{equation}\label{Collaborative_signal} 
    \mathbf{z}_i(k) = w_{ii}(k)\mathbf{y}_i(k)+ \big[\sum_{j\in \mathcal{N}_i}{w_{ij}(k) \mathbf{y}_j}(k) + {\alpha}_i(k)\big ],i\in [1,M]
\end{equation}
where $w_{ij}(k)$ is the $(i,j)$th entry of the collaboration matrix $\textbf{W}(k)$ which means the weight of the observation transmitted from $j$th sensor to the $i$th sensor, $\mathcal{N}_i$ denotes the set of all the neighbors of the $i$th sensor, that is to say, $\mathcal{N}_i = \{j|A_{ij} = 1,j\in [1,N],j\neq i\}$. And $\boldsymbol{\alpha}_i(k)$ is the collaboration noise which is an i.i.d sequence with zero mean and covariance $\textbf{R}_{\boldsymbol{\alpha}_i}$.
The observation at each sensor after collaboration given in \eqref{Collaborative_signal} can be succinctly written as: 
  \begin{equation}
      \mathbf{z}_i(k) = \sum_{j=1}^N{w_{ij}(k) \mathbf{y}_j}(k) + \boldsymbol{\alpha}_i(k),i\in [1,M]
  \end{equation}

We refer to the observations obtained after collaboration as the {\em post-collaboration observations}.

The post-collaboration observations at all the sensors can be compactly expressed as follows.
\begin{equation}
    \mathbf{z}(k) = [\mathbf{W}(k)\otimes \mathbf{I}_L]\mathbf{y}(k) + \boldsymbol{\alpha}(k)
\end{equation}
where $\mathbf{y}(k) = \begin{bmatrix}\mathbf{y}_1^T(k) &\dots & \mathbf{y}_N^T(k)\end{bmatrix}^T\in \mathbb{R}^{NL\times 1}$ are the original observations of all the $N$ sensors, $\mathbf{z}(k) = \begin{bmatrix}\mathbf{z}_1^T(k) & \dots & \mathbf{z}_M^T(k)\end{bmatrix}^T\in \mathbb{R}^{ML\times 1}$ are the post-collaboration observations at the $M$ sensors, $\boldsymbol{\alpha}(k) = \begin{bmatrix}\boldsymbol{\alpha}_1^T(k),\dots,\boldsymbol{\alpha}_M^T(k)\end{bmatrix}^T\in \mathbb{R}^{ML\times 1}$ is the collaboration noise involved in inter-sensor communication, $\mathbf{W}(k)\in \mathcal{R}^{M\times N}$ is the collaboration weight at time $k$,  $\otimes$ denotes the Kronecker product and $\textbf{I}_L$ is the $L$-dimensional identity matrix.

Meanwhile, to further reduce the communication cost between the local sensors and the FC, the observations from sensor nodes with $i \in [1,M]$ are linearly compressed \cite{akhtar2017distributed,zhang2018optimal} as follows:
\begin{equation}
    {z}_i^c(k) = \mathbf{f}_i^T(k)\mathbf{z}_i(k) , i\in [1,M]
\end{equation}
where $\mathbf{f}_i(k)$ is the compression vector at $i$th sensor. 

We refer to the compressed post-collaboration observations as the \textit{compressed observations}. The compressed observations at all the sensors can be compactly written as
\begin{equation}\label{Compressed_signal}
     \mathbf{z}^c(k) 
     =\mathbf{F}(k)[\mathbf{W}(k)\otimes \mathbf{I}_L]\mathbf{y}(k) + \mathbf{F}(k)\boldsymbol{\alpha}(k)
\end{equation}
where $\mathbf{F}(k)= \text{blkdiag} \{\textbf{f}_1^T, \dots, \textbf{f}_M^T\}\in\mathbb{R}^{M\times ML}$ is the compression matrix, and  $\mathbf{z}^c(k) = [ z_1^c(k), \ldots, z_M^c(k)]^T$ represents the compressed observations from all the $M$ sensors.

In order to reduce unnecessary energy cost, in this paper we assume that each sensor is equipped with only one antenna while the FC is equipped with $S$ antennas. Therefore, the signal received by the $j$th antenna at the FC through a coherent MAC can be expressed as:
\begin{equation}
    q_j(k) = \mathbf{g}_j^T(k)\mathbf{z}^c(k) + \epsilon_j(k), j\in [1,S]
\end{equation}
where $\mathbf{g}_j(k)\in \mathbb{R}^{M\times 1}$ is the channel fading between the $M$ sensors and the $j$th antenna of the FC, and $\epsilon_i(k)$ is the corresponding channel noise at the $j$th antenna.

Therefore, the signal received at FC can be summarized as
\begin{equation}\label{q_expression}
\begin{aligned}
        \mathbf{q}(k) =&\mathbf{G}(k)\mathbf{z}^c(k) + \boldsymbol{\epsilon}(k)\\ =&\mathbf{G}(k)\mathbf{F}(k)[\mathbf{W}(k)\otimes \mathbf{I}_L]\mathbf{H}(k)\mathbf{x}\\
        & \qquad + \mathbf{G}(k)\mathbf{F}(k)[\mathbf{W}(k)\otimes \mathbf{I}_L]\mathbf{v}(k) \\
       & \qquad \qquad + \mathbf{G}(k)\mathbf{F}(k)\boldsymbol{\alpha}(k) + \boldsymbol{\epsilon}(k)
\end{aligned}
\end{equation}
where $\mathbf{q}(k) = \begin{bmatrix} q_1(k) \ldots q_S(k)\end{bmatrix}^T \in 
\mathbb{R}^{S\times 1}$ is the received signal at FC through $S$ antennas, $\mathbf{G}(k) = \begin{bmatrix} \mathbf{g}_1(k) \ldots  \mathbf{g}_S(k)  \end{bmatrix}^T \in \mathbb{R}^{S\times M}$ is the channel fading matrix, 
$\mathbf{H}(k) = \begin{bmatrix}\mathbf{H}_1^T(k) \ldots  \mathbf{H}_N^T(k)\end{bmatrix}^T\in \mathbb{R}^{NL\times P}$ is the linear observation matrix, $\mathbf{v}(k) = \begin{bmatrix}\mathbf{v}_1^T(k) \ldots  \mathbf{v}_N^T(k)\end{bmatrix}^T$
and $\epsilon(k)$ are the i.i.d additive Gaussian noise vectors at the FC with zero mean and covariance $\textbf{R}_{\boldsymbol{\epsilon}}=\mathbb{E}\{\boldsymbol{\epsilon}(k)\boldsymbol{\epsilon}^T(k)\}$.

For ease of notation, we define $\mathbf{D}(k) = \mathbf{G}(k)\mathbf{F}(k)[\mathbf{W}(k)\otimes \mathbf{I}_L]\mathbf{H}(k)$ and $\mathbf{n}_q(k) = \mathbf{G}(k)\mathbf{F}(k)[\mathbf{W}(k)\otimes \mathbf{I}_L]\mathbf{v}(k) + \mathbf{G}(k)\mathbf{F}(k)\boldsymbol{\alpha}(k) + \boldsymbol{\epsilon}(k)$, then (\ref{q_expression}) can be rewritten as
 \begin{equation}\label{q_simplified_expression}
        \textbf{q}(k) = \textbf{D}(k)\textbf{x}(k)+ \textbf{n}_q(k)
\end{equation}
Next, we discuss the communication cost associated with collaboration and compression. 

\subsection{Communication Cost }
Consider the energy consumption at each sensor of the entire system. It consists of two parts: inter-sensor communication for spatial collaboration and for transmitting the compressed observations from the selected sensors to the FC.

According to (\ref{Collaborative_signal}), the energy cost consumed by the $i$th sensor at time $k$ can be formulated as:
\begin{equation}\label{Collaboration_cost}
   C_i^{(1)}(k) =  \sum_{j\in{\mathcal{N}}_i}{\| w_{ji}(k)\mathbf{y}_i(k)\| }_2^2   ,i\in [1,N]
\end{equation}
where $\mathcal{N}_i$ denotes the set of all the neighbors of the $i$th sensor.
Given the collaboration matrix and the compression matrix, its expectation over random vector $\mathbf{x}$ can be expressed as 
\begin{equation}\label{collaboration_constraint}
    \begin{split}
        \mathbb{E}[{C}_i^{(1)}(k)] = &  \mathbb{E}[\boldsymbol{y}_i^T(k)\boldsymbol{y}_i(k)]\sum_{j\neq i,j=1}^{M} w_{ji}^2(k)\\
        =& \text{tr}[\mathbf{R}_{\mathbf{y}_i}(k)]\{\mathbf{e}_i^T[\mathbf{W}(k)\odot \widetilde{\mathbf{I}}]^T[\mathbf{W}(k)\odot \widetilde{\mathbf{I}}]\mathbf{e}_i\}\\
    \end{split}
\end{equation}
where $\mathbf{R}_{\mathbf{y}_i}(k) = \mathbf{H}_i(k) \mathbf{R}_x \mathbf{H}_i^T(k) + \mathbf{R}_{\mathbf{v}_i}$, $\text{tr}[\mathbf{R}_{\mathbf{y}_i}(k)]$ is the trace of $\mathbf{R}_{\mathbf{y}_i}(k)$, $\odot$ denotes the Hadamard product, $\widetilde{\mathbf{I}} = \mathbf{1}_M \mathbf{1}_N^T - [\mathbf{I}_M, \mathbf{0}_{M \times (N-M) }]$ and $\mathbf{e}_i$ is a basis vector whose $i$th element is 1, and 0 for others.

Besides, based on (\ref{Compressed_signal}), the energy consumed for transmitting the compressed observations from each sensor to the FC is given by
\begin{equation}\label{compression_cost}
    {C}_i^{(2)}(k)  =\bigg\{\mathbf{f}_i^T\bigg[\sum_{j=1}^N w_{ij}(k)\mathbf{y}_j(k) + \boldsymbol{\alpha}_i(k)\bigg]\bigg\}^2, i\in [1,M]
\end{equation}
and its expectation over $\mathbf{x}$ is written as
\begin{equation}\label{compression_constraint}
    \begin{split}
        \mathbb{E}[C_i^{(2)}(k)]
        =& \mathbf{f}_i^T(k)\mathbf{W}_i(k) \mathbf{R}_{\mathbf{y}}(k)\mathbf{W}_i^T(k) \mathbf{f}_i(k)\\
        &+ \mathbf{f}_i^T(k)(\mathbf{e}_i^T\otimes \mathbf{I}_L)\mathbf{R}_{\boldsymbol{\alpha}}(\mathbf{e}_i\otimes \mathbf{I}_L)\mathbf{f}_i(k)      
    \end{split}
\end{equation}
where $\mathbf{W}_i(k) = \mathbf{e}_i^T\mathbf{W}(k)\otimes \mathbf{I}_L$, $\otimes$ denotes the Kronecker product, $\mathbf{R}_\mathbf{y}$  and $\mathbf{R}_{\boldsymbol{\alpha}}$ are the covariances of the observation and collaboration noises which are given by
\begin{equation}
\mathbf{R}_\mathbf{y}(k) = \begin{bmatrix}
\mathbb{E}[\mathbf{y}_1(k)\mathbf{y}_1^T(k)] & \cdots & \mathbb{E}[\mathbf{y}_1(k)\mathbf{y}_N^T(k)]  \\ 
  \vdots& \ddots  & \vdots \\ 
\mathbb{E}[\mathbf{y}_N(k)\mathbf{y}_1^T(k)] & \cdots & \mathbb{E}[\mathbf{y}_N(k)\mathbf{y}_N^T(k)]
\end{bmatrix}
\end{equation}

\begin{equation}
\mathbf{R}_{\boldsymbol{\alpha}} = \begin{bmatrix}
\mathbb{E}(\boldsymbol{\alpha}_1\boldsymbol{\alpha}_1^T) & \cdots & \mathbb{E}(\boldsymbol{\alpha}_1\boldsymbol{\alpha}_M^T)  \\ 
  \vdots& \ddots  & \vdots \\ 
\mathbb{E}(\boldsymbol{\alpha}_M\boldsymbol{\alpha}_1^T) & \cdots & \mathbb{E}(\boldsymbol{\alpha}_M\boldsymbol{\alpha}_M^T)
\end{bmatrix}
\end{equation}
where 
\begin{equation}
\mathbb{E}[\mathbf{y}_i(k)\mathbf{y}_j^T(k)] = \left\{\begin{matrix}
 \mathbf{H}_i(k) \mathbf{R}_x \textbf{H}_i^T(k) + \mathbf{R}_{v_i},& i=j\\
 \mathbf{H}_i(k) \mathbf{R}_x \mathbf{H}_j^T(k),& i\neq j
\end{matrix}\right.
\end{equation}

\begin{equation}
\mathbb{E}(\boldsymbol{\alpha}_i\boldsymbol{\alpha}_j^T) = \left\{\begin{matrix}
\mathbf{R}_{\alpha_i},& i=j\\
 0,& i\neq j
\end{matrix}\right.
\end{equation}

In short, the overall energy consumed at the $i$th sensor could be expressed as:
\begin{equation}\label{Energy constraint}
    \mathbb{E}[{C}_i(k)] = \left\{\begin{array}{cl}
\mathbb{E}[{C}_i^{(1)}(k)] +\mathbb{E}[{C}_i^{(2)}(k)] , &i\in [1,M]\\
\mathbb{E}[{C}_i^{(1)}(k)],& i\in [M+1,N]
\end{array}\right.
\end{equation}

Before concluding this part, some intuition behind energy allocation  is discussed. For the first M sensors, there is a trade-off  among collaboration and compression. To obtain better estimation performance, those sensors with poor channel quality associated with the FC may prefer to allocate energy to collaboration instead of compression. On the contrary, those sensors with better channel quality prefer to assign most of their energy for compression.

\subsection{Problem Formulation}
In this part, a recursive estimator of $\mathbf{x}$ is presented under the proposed framework. Let all the previous observations at time $k$ be denoted as $\Theta(k)= \{\mathbf{q}(0),\mathbf{q}(1),\dots,\mathbf{q}(k)\}$, then, the optimal recursive estimator at time $k$ is given by 
\begin{equation}\label{x_update}
    \begin{split}
        \hat{\mathbf{x}}(k) =& \mathbb{E}[\mathbf{x}|\Theta(k)] \\
        =& \hat{\mathbf{x}}(k-1) +\mathbf{T}(k) \{\mathbf{q}(k)-\mathbb{E}[\mathbf{q}(k)|\Theta(k-1) \}
    \end{split}
\end{equation}
where $\mathbf{T}(k)$ is the filter gain at FC.

As the noise at time $k$ is independent of all the previous observations at time $k-1$, namely $\Theta(k-1)$, then 
\begin{equation}
    \mathbb{E}[\mathbf{q}(k)|\Theta(k-1)] = \mathbf{D}(k)\hat{\mathbf{x}}(k-1)
\end{equation}

Consequently, the estimation error can be expressed as
\begin{equation}\label{Estimation_error}
\begin{split}
    \mathbf{e}(k) =&  \mathbf{x}(k) - \hat{\mathbf{x}}(k)\\
    =&[\mathbf{I} - \mathbf{T}(k)\mathbf{D}(k)]\mathbf{e}(k-1) -\mathbf{T}(k)\mathbf{n}_q(k)
\end{split}
\end{equation}

Correspondingly, the error covariance is given by
\begin{equation}\label{error_covariance}
\begin{split}
        \mathbf{P}(k)  = &[\mathbf{I} - \mathbf{T}(k)\mathbf{D}(k)]\mathbf{P}(k-1)[\mathbf{I} - \mathbf{T}(k)\mathbf{D}(k)]^T\\
        &+ \mathbf{T}(k)\mathbf{R}_n(k)\mathbf{T}^T(k)
\end{split}
\end{equation}
where  
$\mathbf{R}_n(k) =\mathbf{G}(k)\mathbf{F}(k)[\mathbf{W}(k)\otimes \mathbf{I}_L]\mathbf{R}_v[\mathbf{W}(k)\otimes \mathbf{I}_L]^T\mathbf{F}^T(k)\mathbf{G}^T(k) +\mathbf{G}(k)\mathbf{F}(k)\mathbf{R}_{\boldsymbol{\alpha}}\mathbf{F}^T(k)\mathbf{G}^T(k) +   \mathbf{R}_{\epsilon}$ is the noise covariance and 

\begin{equation}
\mathbf{R}_{\mathbf{v}} = \begin{bmatrix}
\mathbb{E}(\mathbf{v}_1\mathbf{v}_1^T) & \cdots & \mathbb{E}(\mathbf{v}_1\mathbf{v}_N^T)  \\ 
  \vdots& \ddots  & \vdots \\ 
\mathbb{E}(\mathbf{v}_N\mathbf{v}_1^T) & \cdots & \mathbb{E}(\mathbf{v}_N\mathbf{v}_N^T)
\end{bmatrix}
\end{equation}
where 
\begin{equation}
\mathbb{E}(\mathbf{v}_i\textbf{v}_j^T) = \left\{\begin{matrix}
 \mathbf{R}_{\mathbf{v}_i},& i=j\\
0,& i\neq j
\end{matrix}\right.
\end{equation}

The Mean square error (MSE) could be written as 
$ \Phi\{\mathbf{W}(k),\mathbf{F}(k),\mathbf{T}(k) \}=\mathbb{E}[\textbf{e}^T(k)\textbf{e}(k)]=\text{tr}[\mathbf{P}(k)]$. Therefore, our problem could be summarized as 

\begin{equation}\label{Target_function}
\begin{aligned}
        \mathop{\text{minimize}}_{\mathbf{T},\mathbf{W},\mathbf{F}}   \quad \  &\Phi\{\mathbf{T}(k),\mathbf{W}(k),\mathbf{F(k)} \}\\
        \text{subject to} \ \quad &     \mathbf{W}(k) \odot (\mathbf{1}_M \mathbf{1}_N^T - \mathbf{A}) = \mathbf{0}\\
        & \mathbb{E}(C_i)\leq \mu_i, i\in [1,N]\\
\end{aligned}
\end{equation}

\section{Optimal collaboration and compression matrix design}
\label{Sec: Optimal collaboration and compression matrix design}
In this section, a detailed algorithm is introduced for implementing the optimal estimation scheme under the proposed framework. 
To achieve this, an iterative optimization method is used in this paper. In Sec \ref{Sec-A: optimal collaboraiton} to \ref{Sec-C: dcentralized compression}, we obtain the optimal solutions for $\mathbf{W}(k)$,$\mathbf{F}(k)$,$\mathbf{T}(k)$ which yield the centralized and decentralized algorithms respectively. Finally, the convergence of the proposed algorithm is investigated.

\subsection{Optimal Sensor Collaboration}\label{Sec-A: optimal collaboraiton}
As the structure of $\mathbf{W}(k)$ may be sparse which is determined by the sensor network topology, it makes the problem (\ref{Target_function}) hard to solve. 
To overcome this problem, motivated by \cite{zhang2018optimal}, we establish the correspondence between the structure of the sparse network topology and collaboration weights. 
Specifically, we first vectorize the collaboration matrix and eliminate all the elements whose corresponding entry in the network topology matrix $\mathbf{A}$ equals to zero, then constitute a new vector $\mathbf{w}\in \mathbb{R}^{U\times 1}$ where $U$ is the total number of nonzero elements in $\mathbf{A}$. Obviously, there exists a unique one-to-one mapping from $\mathbf{W}$ to $\mathbf{w}$ which can be expressed as
\begin{equation}\label{map_w_W}
    \mathbf{w}_u = [\mathbf{W}]_{m_u n_u}
\end{equation}
where $m_u$ and $n_u$ are the corresponding row and column indices of the $u$th entry of vector $\mathbf{w}$ and $[\mathbf{W}]_{mn}$ denotes the $(m,n)$th element of matrix $\mathbf{W}$.

By observing the expanded expression of the trace of the error covariance in (\ref{error_covariance_expression}), it can be seen that the error covariance consists of two specific functions in terms of collaboration matrix $\mathbf{W}(k)$, which are $\text{tr}[\mathbf{B}(\mathbf{W}\otimes \mathbf{I})\mathbf{C}(\mathbf{W}\otimes \mathbf{I})^T\mathbf{D}]$ and $\text{tr}[\mathbf{B}(\mathbf{W}\otimes \mathbf{I}_L)\mathbf{C}]$.
In order to simplify the original problem (\ref{Target_function}), a relationship between $\mathbf{W}$ and $\mathbf{w}$ is observed and provided in Proposition 1.

\begin{prop}
    Given a matrix $\mathbf{W}\in \mathbb{R}^{M\times N}$ and its columnwise vector $\mathbf{w} \in \mathbb{R}^U$ which consists of the nonzero element of $\mathbf{W}$ such that $w_u = W_{m_un_u}$ where $m_u\in[1,M]$,$n_u\in [1,N]$,$u\in [1,U]$. The expression of $\mathbf{a}^T(\mathbf{W}\otimes \mathbf{I}_L)$, $\text{tr}[\mathbf{B}(\mathbf{W}\otimes \mathbf{I}_L)\mathbf{C}(\mathbf{W}\otimes \mathbf{I}_L)^T\mathbf{D}]$, and $\text{tr}[\mathbf{B}(\mathbf{W}\otimes \mathbf{I}_L)\mathbf{C}]$ can be equivalently expressed as functions of $\mathbf{w}$.
    \begin{equation}\label{map_property1}
        \mathbf{a}^T(\mathbf{W}\otimes \mathbf{I}_L) = \mathbf{w}^T \mathbf{A}
    \end{equation}
    \begin{equation}\label{map_property2}
        \text{tr}[\mathbf{B}(\mathbf{W}\otimes \mathbf{I}_L)\mathbf{C}(\textbf{W}\otimes \mathbf{I}_L)^T\mathbf{D}] = \mathbf{w}^T \mathbf{E} \mathbf{w}
    \end{equation}
    \begin{equation}\label{map_property3}
        \text{tr}[\mathbf{B}(\mathbf{W}\otimes \mathbf{I}_L)\mathbf{C}] = \mathbf{w}^T \widetilde{\mathbf{c}}
    \end{equation}
where $\textbf{A}\in \mathbb{R}^{U\times (NL)}$, $\textbf{E}\in \mathbb{R}^{U\times U}$ and $\widetilde{\mathbf{c}}$ are given by 
\begin{equation}
    [\mathbf{A}]_{uj} = \left\{\begin{matrix}
 \mathbf{a}_{L(m_u-1)+ j-L\lfloor{\frac{j-1}{L}}\rfloor},& n_u = \lfloor{\frac{j-1}{L}}\rfloor+1\\ 
 0 ,&\text{otherwise}
\end{matrix}\right.
\end{equation} 
\begin{equation}\label{expression_E}
  \mathbf{E} = \sum_{i=1}\widetilde{\mathbf{B}}_i \mathbf{C} \widetilde{\mathbf{D}}_i^T  
\end{equation}
\begin{equation}
    \widetilde{\mathbf{c}}= \sum_{i=1}\widetilde{\mathbf{B}}_i \mathbf{c}_i
\end{equation}
Also, $\widetilde{\mathbf{B}}_i $ and $\widetilde{\mathbf{C}}_i$ are given by
\begin{equation}
    [\widetilde{\mathbf{B}}_i]_{uj}=\left\{\begin{matrix}
 [\mathbf{b}_i]_{L(m_u - 1)+ j-L\lfloor{\frac{j-1}{L}}\rfloor },& n_u=\lfloor{\frac{j-1}{L}}\rfloor +1 \\ 
 0,& \text{otherwise} 
\end{matrix}\right.
\end{equation}
\begin{equation}
[\widetilde{\textbf{D}}_i]_{uj}=\left\{\begin{matrix}
 [\textbf{d}_i]_{L(m_u - 1)+ j-L\lfloor{\frac{j-1}{L}}\rfloor },& n_u=\lfloor{\frac{j-1}{L}}\rfloor +1 \\ 
 0,& \text{otherwise} 
\end{matrix}\right.
\end{equation}
for $u\in [1,U]$ and $j\in [1,NL]$, where $\mathbf{b}_i$ is the $i$th row of matrix $\mathbf{B}$, $\mathbf{c}_i$ is the $i$th column of matrix $\mathbf{C}$ and $\mathbf{d}_i$ is the $i$th column of matrix $\mathbf{D}$.
\end{prop}
\begin{proof}
    See Appendix \ref{Appendix: Constraint function}.
\end{proof}

Based on Proposition 1, the original problem in terms of $\mathbf{W}(k)$ can be reformulated as the function of $\mathbf{w}$
\begin{equation}\label{Solution of w}
\begin{aligned}
        \mathop{\text{minimize}}_{\mathbf{w}}  \quad \  &\textbf{w}^T\boldsymbol{\Omega}^{(0)}\mathbf{w} - 2\mathbf{w}^T\mathbf{d}+ \eta_0\\
         \text{subject to} \quad & \mathbf{w}^T \boldsymbol{\Omega}_i^{(1)} \textbf{w}+\textbf{w}^T \boldsymbol{\Omega}_i^{(2)} \mathbf{w} + \eta_i \leq \mu_i, i\in [1,M]\\
          &\mathbf{w}^T \boldsymbol{\Omega}_i^{(1)} \mathbf{w}\leq \mu_i, i\in [M+1,N]\\
\end{aligned}
\end{equation}
where $\boldsymbol{\Omega}_i^{(0)}$,$\boldsymbol{\Omega}_i^{(1)}$ and $\boldsymbol{\Omega}_i^{(2)}$ are all positive definite matrices \footnote{If $\mathbf{r}^T\mathbf{A}\mathbf{r}>0$ holds for any nonzero vector $\mathbf{r}$, then the symmetric real matrix $\mathbf{A}$ is called positive definite matrix. Normally, it is denoted as $\mathbf{A} \succ 0$. Moreover, if $\mathbf{r}^T\mathbf{A}\mathbf{r}\geq 0$ holds for any vector $\mathbf{r}$, then the symmetric real matrix $\mathbf{A}$ is called positive semi-definite matrix. And it is denoted as $\mathbf{A} \succeq 0$.}. The expressions of these coefficient matrices and the proof of positive definiteness are both given in Appendix \ref{Appendix: Coefficient matrix}. 

The problem formulated in (\ref{Solution of w}) is a {\it{Quadratically Constrained Quadratic Programming}} (QCQP) problem which is an NP-hard problem in general. However, the coefficient $\boldsymbol{\Omega}_i^{(0)}$,$\boldsymbol{\Omega}_i^{(1)}$ and $\boldsymbol{\Omega}_i^{(2)}$ in (\ref{Solution of w}) are both positive definite, which means it is a convex QCQP problem. Therefore, it can be directly solved by using interior-point methods with standard solvers \cite{nemirovski2004interior}. 

By solving the QCQP problem in (\ref{Solution of w}), the optimal collaboration matrix is obtained. Once the collaboration matrix is determined, sensors share their weighted observations with their neighbors. 
These post-collaboration observations at the $M$ sensors are first compressed before transmission to the FC. 
Then the $M$ sensors send the compressed observations to the FC. 
This compression is designed to meet the power constraints at individual nodes as well as to reduce the communication overhead. Next, we design near-optimal compression strategies for the centralized and the decentralized cases. 

\subsection{Optimal Compression and Filter gain: Centralized compression case}\label{Sec-B: centralized compression}
In this section, we first consider the centralized case under the proposed framework. Note that, the compression matrix $\mathbf{F}(k)$ is a block diagonal matrix, which can not be optimized directly. However, if we rewrite the problem (\ref{Target_function}) in terms of $\mathbf{f}_i(k)$, the target function in (\ref{Target_function}) can be expressed as 
\begin{equation}\label{target_fun_f_cen}
    \begin{split}
        &\Phi\{\mathbf{T}(k),\mathbf{W}(k),\mathbf{F(k)} \}\\=&
         \text{tr}[\mathbf{P}(k-1)]+ \text{tr}[\mathbf{T}(k)\mathbf{R}_{\boldsymbol{\epsilon}} \mathbf{T}^T(k)]\\
        &+ \sum_{i=1}^{M}\mathbf{f}_i^T(k)\mathbf{W}_i(k)\mathbf{H}\mathbf{P}(k-1)\mathbf{H}^T\mathbf{W}_i^T(k)\mathbf{f}_i(k)\pi_{ii}\\
        &-2\sum_{i=1}^{M}\mathbf{f}_i^T(k)\mathbf{W}_i(k)\mathbf{H}\mathbf{P}(k-1)\mathbf{T}(k)\mathbf{g}_i \\
        &+ \sum_{i=1}^{M}\mathbf{f}_i^T(k)[\mathbf{W}_i(k)\mathbf{R}_{\mathbf{v}}\mathbf{W}_i^T(k)+\mathbf{R}_{\alpha_i} ]\mathbf{f}_i(k)\pi_{ii}\\
        &+ 2\sum_{\mathbb{S}^1}\mathbf{f}_i^T(k)\mathbf{W}_i(k)\mathbf{H}\mathbf{P}(k-1)\mathbf{H}^T\mathbf{W}_i^T(k)\mathbf{f}_j(k)\pi_{ji}\\
        &+ 2 \sum_{\mathbb{S}^1}\mathbf{f}_i^T(k)\mathbf{W}_i(k)\mathbf{R}_{\mathbf{v}}\mathbf{W}_i^T(k) \mathbf{f}_j(k)\pi_{ji}\\
        \end{split}
\end{equation}
where $\pi_{ij}=\mathbf{g}_i^T(k)\mathbf{T}^T(k)\mathbf{T}(k)\mathbf{g}_j(k)$ for $i\in [1,M]$, $j\in [1,N]$ and $\mathbb{S}^1=\{(i,j)|i\in [1,M],j\in [1,N] , i\neq j\}$. 

For simplicity,  the target function could be equivalently expressed as the function with respect to $\mathbf{f}_i(k)$
\begin{equation}\label{target_fun_fi}
    \begin{split}
        &\Upsilon_c\{\mathbf{f}_i(k),\mathbf{f}_j(k) \}\\=&
         \mathbf{f}_i^T(k)\mathbf{W}_i(k)\mathbf{H}(k)\mathbf{P}(k-1)\mathbf{H}^T(k)\mathbf{W}_i^T(k)\mathbf{f}_i(k)\pi_{ii}\\
        &-2 \mathbf{f}_i^T(k)\mathbf{W}_i(k)\mathbf{H}(k)\mathbf{P}(k-1)\mathbf{T}(k)\mathbf{g}_i \\
        &+ \mathbf{f}_i^T(k)[\mathbf{W}_i(k)\mathbf{R}_{\mathbf{v}}\mathbf{W}_i^T(k)+\mathbf{R}_{\alpha_i} ]\mathbf{f}_i(k)\pi_{ii}\\
        &+ 2\sum_{\mathbb{S}^2}\mathbf{f}_i^T(k)\mathbf{W}_i(k)\mathbf{H}\mathbf{P}(k-1)\mathbf{H}^T\mathbf{W}_i^T(k)\mathbf{f}_j(k)\pi_{ji}\\
        &+ 2\sum_{\mathbb{S}^2} \mathbf{f}_i^T(k)\mathbf{W}_i(k)\mathbf{R}_{\mathbf{v}}\mathbf{W}_i^T(k) \mathbf{f}_j(k)\pi_{ji}\\
        =&\mathbf{f}_i^T(k)\boldsymbol{\Omega}_i^{(3)}\mathbf{f}_i(k) -2\mathbf{f}_i(k)\mathbf{d}^{(1)}\\
    \end{split}
\end{equation}
where 
\begin{equation}
\begin{split}
      \boldsymbol{\Omega}^{(3)}=&\pi_{ii}\big[\mathbf{W}_i(k)\mathbf{H}(k)\mathbf{P}(k-1)\mathbf{H}^T(k)\mathbf{W}_i^T(k)\\
      & + \mathbf{W}_i(k)\mathbf{R}_{\mathbf{v}}\mathbf{W}_i^T(k)+\mathbf{R}_{\alpha_i}\big]\succeq 0  
\end{split}
\end{equation}
\begin{equation}
    \begin{split}
        \mathbf{d}^{(1)} =& \mathbf{W}_i(k)\mathbf{H}(k)\mathbf{P}(k-1)\mathbf{T}(k)\mathbf{g}_i\\
        &+\sum_{\mathbb{S}^2}\mathbf{W}_i(k)\mathbf{H}\mathbf{P}(k-1)\mathbf{H}^T\mathbf{W}_i^T(k)\mathbf{f}_j(k)\pi_{ji}\\
        &+\sum_{\mathbb{S}^2} \mathbf{W}_i(k)\mathbf{R}_{\mathbf{v}}\mathbf{W}_i^T(k) \mathbf{f}_j(k)\pi_{ji}
    \end{split}
\end{equation}
and where the matrix $\boldsymbol{\Omega}^{(3)}$ is positive semi-definite (please see Appendix \ref{Appendix: Coefficient matrix})
and we have $\mathbb{S}^2=\{j|j\in [1,N], i \in [1,M], j\neq i\}$.

From (\ref{target_fun_fi}), it can be seen that the optimization in terms of $\mathbf{f}_i$ depends on $\mathbf{f}_j$ for $j\neq i$. In order to obtain the optimal solution of $f_i$ for all $i\in [1,M]$, an alternative method is used here. At each iteration, each $\mathbf{f}_i(k)$ is optimized by solving the following problem
\begin{equation}\label{fi_central}
    \begin{split}
      \mathop{\text{minimize}}_{\mathbf{f}_i^t(k)} \quad &\Upsilon_c[\mathbf{f}_i^t(k),\mathbf{f}_j^{t-1}(k)] \\
        \text{subject to}\quad   & [\mathbf{f}_i^t(k)]^T\boldsymbol{\Omega}_i^{(4)}\mathbf{f}_i^t(k) + \lambda_i \leq \mu_i, i\in [1,M]\\
    \end{split}
\end{equation}
where $\mathbf{f}_i^t(k)$ denotes the $t$th iteration for $\mathbf{f}_i(k)$, $\boldsymbol{\Omega}_i^{(4)}$ and $\lambda_i$ are coefficients which are given by
\begin{equation}
    \boldsymbol{\Omega}_i^{(4)}= \mathbf{W}_i(k) \mathbf{R}_{\mathbf{Y}}\mathbf{W}_i^T(k) + (\mathbf{e}_i^T\otimes \mathbf{I}_L)\mathbf{R}_{\boldsymbol{\alpha}}(\mathbf{e}_i\otimes \mathbf{I}_L)    
\end{equation}
\begin{equation}
    \lambda_i = \text{tr}(\mathbf{R}_{\mathbf{y}_i})\{\mathbf{e}_i^T[\mathbf{W}(k)\odot \widetilde{\mathbf{I}}]^T[\mathbf{W}(k)\odot \widetilde{\mathbf{I}}]\mathbf{e}_i\}
\end{equation}
where $\boldsymbol{\Omega}_i^{(4)}$ is positive definite.
As the problem in (\ref{fi_central}) is also a convex QCQP problem, it could be solved similarly as (\ref{Solution of w}) by using interior-point methods with standard solvers.

Once the collaboration matrix $\mathbf{W}(k)$ and the compression matrix $\mathbf{F}(k)$ are obtained, the closed form of the filter gain $\mathbf{T}(k)$ is given by 
\begin{equation}\label{Closed_form_T}
\begin{split}
     \mathbf{T}(k)=& \mathbf{P}(k-1)\mathbf{H}^T[\mathbf{W}(k)\otimes \mathbf{I}]^T \mathbf{F}^T(k) \mathbf{G}^T \\ &[\mathbf{D}(k)  \mathbf{P}(k-1)
      \mathbf{D}^T(k) 
      +\mathbf{R}_n(k) ]^{-1}
\end{split}
\end{equation}

Then, the centralized algorithm for solving problem (\ref{Target_function}) is implemented.
The centralized estimation framework is detailed below.
\begin{enumerate}
    \item Each sensor estimates the observation matrices $\mathbf{H}_i(k)$  by using a pilot-based method \cite{tong2004pilot} which is transmitted to the neighbors that are tasked to communicate with the FC. 
    Then all these matrices are transmitted to the FC. Meanwhile, the FC also estimates the channel matrix $\mathbf{G}(k)$ by using the same technique.
    \item  At each time instant $k$, the collaboration matrix $\mathbf{W}(k)$, the compression matrix $\mathbf{f}_i(k)$ and the filter gain $\mathbf{T}(k)$ are optimized at the FC by solving (\ref{Solution of w}), (\ref{fi_central}), and (\ref{Closed_form_T}) separately.
    \item The FC broadcasts $\mathbf{W}(k)$ to all the $N$ sensors and $\mathbf{f}_i(k)$ to the $M$ sensors that are tasked to communicate with the FC.
    \item Then, the post-collaboration observations are compressed using (\ref{Compressed_signal}), and the compressed observations are transmitted to the FC.
\end{enumerate}

By observing (\ref{target_fun_f_cen}), one can see that the communication costs can be high if the observation matrices $\mathbf{H}(k)$ change quiet frequently, i.e., the channels are fast fading. However, if the coherence interval of the observation matrices  spans over multiple time instants, i.e., channels are slow fading \cite{xiao2008linear}, then the communication costs will be acceptable. 
In order to distribute some of the computational load of the FC, one can alternatively choose to design the compression matrices locally at the individual sensors. This algorithm is referred to as the decentralized compression case and is discussed next. 

\subsection{Optimal Compression and Filter gain: Decentralized compression case}\label{Sec-C: dcentralized compression}
In order to obtain compression vectors $\mathbf{f}_i(k)$ locally, the solution of $\mathbf{f}_i(k)$ can not depend on $\mathbf{f}_j(k)$ for $j \neq i$. Define
$\Lambda_i = \sum_{j\neq i}\textbf{g}_i\textbf{f}_i^T(k)\mathbf{W}_i(k)[\textbf{R}_{\textbf{v}}+\textbf{H}\textbf{P}(k-1)\textbf{H}^T]\mathbf{W}_j^T(k) \textbf{f}_j(k)\textbf{g}_j^T$, then $\mathbf{f}_i(k)$ could be obtained locally as long as the following condition holds 
\begin{equation}\label{condition_T}
    \text{tr}\Big\{\mathbf{T}(k)\Lambda_i(k)\mathbf{T}^T(k)\Big\}=0
\end{equation}
Using the above condition, the target function in (\ref{target_fun_fi}) becomes
\begin{equation}\label{target_f_decen}
    \Upsilon_d[\mathbf{f}_i(k)]= \mathbf{f}_i^T(k)\boldsymbol{\Omega}_i^{(3)}\mathbf{f}_i(k) -2\mathbf{f}_i^T(k)\mathbf{d}^{(2)}
\end{equation}
where $\mathbf{d}^{(2)} = \mathbf{W}_i(k)\mathbf{H}\mathbf{P}(k-1)\mathbf{T}(k)\mathbf{g}_i$.

Then, the optimization problem in terms of $\mathbf{f}_i(k)$ becomes:
\begin{equation}\label{decentralized_f}
    \begin{split}
      \mathop{\text{minimize}}_{\mathbf{f}_i(k)} \quad &\Upsilon_d[\mathbf{f}_i(k)] \\
        \text{subject to}\quad   & \mathbf{f}_i^T(k)\boldsymbol{\Omega}_i^{(4)}\mathbf{f}_i(k) + \lambda_i \leq \mu_i, i\in [1,M]\\
    \end{split}
\end{equation}

Clearly, (\ref{decentralized_f}) does not depend on the information from the other channels which enables each sensor to obtain their individual compression vectors locally. 
However, to ensure that the condition (\ref{condition_T}) holds, the solution of filter gain $\mathbf{T}(k)$ becomes a non-linear constrained problem as 
\begin{equation}\label{original_function_T}
\begin{aligned}
        \mathop{\text{minimize}}_{\mathbf{T}(k)}   \quad \  &\Phi\{\mathbf{T}(k),\mathbf{W}(k),\mathbf{F(k)} \}\\
        \text{subject to} \ \quad &         \text{tr}\Big\{\mathbf{T}(k)\Lambda_i(k)\mathbf{T}^T(k)\Big\}=0
\end{aligned}
\end{equation}
To solve this problem (\ref{original_function_T}), we reformulate it as: Let $\widetilde{\mathbf{t}}(k)= \text{vec}[\mathbf{T}(k)]$, where $\widetilde{\mathbf{t}}(k)$ is the vectorized form of $\mathbf{T}$. Then (\ref{original_function_T}) can be transformed as
\begin{equation}\label{Dc_T_optimization}
\begin{aligned}
    \mathop{\text{minimize}}_{\widetilde{\mathbf{t}}}  \quad \  &  \widetilde{\mathbf{t}}^T \boldsymbol{\Omega}_T^{(1)} \widetilde{\mathbf{t}}  -2 \widetilde{\boldsymbol{\ell}}^T \widetilde{\mathbf{t}} \\
   \text{subject to}\quad   & \widetilde{\mathbf{t}}^T \boldsymbol{\Omega}_T^{(2)} \widetilde{\mathbf{t}} = 0\\
\end{aligned}
\end{equation}
where $\boldsymbol{\Omega}_T^{(1)}$, $\boldsymbol{\Omega}_T^{(2)}$ and $\widetilde{\boldsymbol{\ell}}$ are given by 
\begin{equation}
    \begin{aligned}
       \boldsymbol{\Omega}_T^{(1)} = &\big \{ \mathbf{D}(k)\mathbf{P}(k-1)\mathbf{D}^T(k) \\  
       & + \mathbf{G}\mathbf{F}(k)[\mathbf{W}(k)\otimes\mathbf{I}_L]\mathbf{R}_{\mathbf{v}}[\mathbf{W}(k)\otimes\mathbf{I}_L]^T\mathbf{F}^T(k)\mathbf{G}^T\\
       & + \mathbf{G}\mathbf{F}(k) \mathbf{R}_{\alpha}\mathbf{F}^T(k)\mathbf{G}^T   + \mathbf{R}_{\boldsymbol{\epsilon}} \big \}\otimes \mathbf{I}_P
    \end{aligned}
\end{equation}
\begin{equation}\label{omega_T_2}
    \boldsymbol{\Omega}_T^{(2)} = \mathbf{B}^T(k) \otimes \mathbf{I}_P
\end{equation}
\begin{equation}
    \widetilde{\boldsymbol{\ell}} = \text{vec}\{\mathbf{P}(k-1)\mathbf{H}^T[\mathbf{W}(k)\otimes\mathbf{I}_L]^T\mathbf{F}^T(k)\mathbf{G}^T\} 
\end{equation}
where $\mathbf{B}(k) = \sum_{i=1}^{M}\boldsymbol{\Lambda}_i$.

Even though $\boldsymbol{\Omega}_T^{(1)}$ is a positive definite matrix which means the target function in (\ref{Dc_T_optimization}) is convex, $\boldsymbol{\Omega}_T^{(2)}$ is not a positive definite matrix. Thus, problem (\ref{Dc_T_optimization}) turns out to be a non-convex QCQP which is hard to solve in general. 
However, there is only one constraint in problem (\ref{Dc_T_optimization}). Motivated by the strategy proposed in \cite{park2017general} for solving one-constraint QCQP problems, the problem (\ref{Dc_T_optimization}) can be solved by making use of the symmetry of matrix $\boldsymbol{\Omega}_T^{(2)}$.


Due to the fact that $\boldsymbol{\Omega}_T^{(1)}$ is a positive definite matrix, its eigenvalue decomposition is given by $\boldsymbol{\Omega}_T^{(1)}= \mathbf{U}_{1}\boldsymbol{\Sigma}_{1}\mathbf{U}_{1}^T$ and $\boldsymbol{\Sigma}_1=\text{diag}(\delta_1,\delta_2,\dots,\delta_{PS})$ where $\delta_i>0$ is the $i$th eigenvalue of $\boldsymbol{\Omega}_T^{(1)}$. Let $\mathbf{V} = \mathbf{U}_1 \boldsymbol{\Sigma}_1^{1/2}$, then 
\begin{equation}
    \mathbf{V}^{-1}\boldsymbol{\Omega}_T^{(1)}\mathbf{V}^{-T} = \boldsymbol{\Sigma}_1^{-1/2}\mathbf{U}_1^T\boldsymbol{\Omega}_T^{(1)}\mathbf{U}_1 \boldsymbol{\Sigma}_1^{-1/2}=\mathbf{I}_{PS}
\end{equation}

Notice that $\boldsymbol{\Omega}_T^{(2)}$ is real symmetric then $\mathbf{V}^{-1}\boldsymbol{\Omega}_T^{(2)}\mathbf{V}^{-T}$ is also real symmetric which means it can be diagonalized as $\mathbf{V}^{-1}\boldsymbol{\Omega}_T^{(2)}\mathbf{V}^{-T} = \mathbf{U}_{2}\boldsymbol{\Sigma}_{2}\mathbf{U}_{2}^T$. Let $\mathbf{M}_0 = \mathbf{V}\mathbf{U}_2$, then
\begin{equation}
    \mathbf{M}_0^{-1}\boldsymbol{\Omega}_T^{(2)}\mathbf{M}_0^{-T} = \mathbf{U}_2^T\mathbf{V}^{-1}\boldsymbol{\Omega}_T^{(2)}\mathbf{V}^{-T}\mathbf{U}_2 = \boldsymbol{\Sigma}_{2}
\end{equation}

Now, let $\mathbf{M} = \mathbf{M}_0^{-1}$, it can be shown that 
\begin{equation}
    \mathbf{M}\boldsymbol{\Omega}_T^{(1)}\mathbf{M}^T =\mathbf{I}_{PS},\quad \mathbf{M}\boldsymbol{\Omega}_T^{(2)}\mathbf{M}^T =\boldsymbol{\Sigma}_{2}
\end{equation}
where $\boldsymbol{\Sigma}_2 = \text{diag}(\sigma_1,\dots,\sigma_{PS})$, and $\sigma_i$ is the $i$th eigenvalue of $\boldsymbol{\Omega}_T^{(2)}$.

Then, the problem in (\ref{Dc_T_optimization}) can be rewritten as
\begin{equation}\label{T_expression1}
 \begin{aligned}
    \mathop{\text{minimize}}_{{\mathbf{r}}}  \quad \  &  \mathbf{r}^T \mathbf{r}  -2 \mathbf{r}^T \mathbf{M}\widetilde{\boldsymbol{\ell}} \\
   \text{subject to}\quad   & \mathbf{r}^T \boldsymbol{\Sigma}_2\mathbf{r} = 0\\
\end{aligned}   
\end{equation}
where $\mathbf{r} = \mathbf{M}^{-T}\widetilde{\mathbf{t}}$.

The Lagrangian of problem (\ref{Dc_T_optimization}) is given by 
\begin{equation}\label{T_Lagrangian}
    L(\mathbf{r},\beta) = \mathbf{r}^T(\mathbf{I} + \beta \boldsymbol{\Sigma}_2)\mathbf{r} -2 \mathbf{r}^T \mathbf{M}\widetilde{\boldsymbol{\ell}}
\end{equation}
Since $\boldsymbol{\Omega}_T^{(2)}$ is non-positive definite, there exists a feasible $\beta$ that could satisfy $\mathbf{I} + \beta \boldsymbol{\Sigma}_2 \succeq 0$. Then, there are two cases:
\begin{enumerate}
    \item Case one: $\mathbf{I} + \beta \boldsymbol{\Sigma}_2 \succ 0$
    
    Notice the range of $\beta$ which satisfies $\mathbf{I} + \beta \boldsymbol{\Sigma}_2 \succ 0$, that is to say, $1 + \beta \sigma_i > 0$ for all $i$. In this range, we can find the minimum value of problem (\ref{T_Lagrangian}) by taking the derivative with respect to $\mathbf{r}$ and letting it equal to zero:
    \begin{equation}\label{r_solution}
        \mathbf{r} = (\mathbf{I}+\beta \Sigma_2)^{-1}\mathbf{M}\boldsymbol{\widetilde{\ell}}
    \end{equation}
    
    Substitute $\mathbf{r}$ into the equality constraint in (\ref{T_expression1}), and let $\mathbf{m} = \mathbf{M}\boldsymbol{\widetilde{\ell}}$, then we could get a nonlinear equation with respect to $\beta$
    \begin{equation}\label{beta_solution}
        \sum_i^{PS}\frac{\sigma_i m_i^2}{(1+\beta \sigma_i)^2}=0
    \end{equation}
    where $m_i$ is the $i$th element of $\mathbf{m}$.
    Therefore, as long as the solution obtained from (\ref{beta_solution}) belongs to the range of $\mathbf{I}+ \beta \boldsymbol{\Sigma}_2\succ 0$, then the corresponding $\mathbf{r}$ in (\ref{r_solution}) is the optimal solution of problem (\ref{T_expression1}). To obtain the solution of (\ref{beta_solution}), notice that the derivative of the lefthand side in (\ref{beta_solution}) with respect to $\beta$ is 
    \begin{equation}
        -\sum_i^{PS} \frac{2\sigma_i^2 m_i^2}{(1+\beta \sigma_i)^3}<0
    \end{equation}
    which means the lefthand side of (\ref{beta_solution}) monotonically decreases with increasing $\beta$. Therefore, we can find the solution by looking for where the change of sign in the lefthand happens using the bisection method.
    \item Case two: $\mathbf{I} + \beta \boldsymbol{\Sigma}_2 \succeq 0$ and $\mathbf{I} + \beta \boldsymbol{\Sigma}_2$ is singular. As $\boldsymbol{\Omega}_T^{(2)}$ is indefinite, there are two solutions of $\beta$ that are possible. One is that $\beta = -{1}/{\sigma_{min}}$ when $\sigma_{min} < 0$, the other one is that $\beta = -1/\sigma_{max}$ when $\sigma_{max} >0$. Then, check if there is any $\mathbf{r}$ that could make the Karush-Kuhn-Tucker (KKT) conditions hold for these $\beta$:
    \begin{equation}
        (\mathbf{I}+\beta \boldsymbol{\Sigma}_2)\mathbf{r} = -\mathbf{M}\boldsymbol{\widetilde{\ell}},\quad
        \mathbf{r}^T \boldsymbol{\Sigma}_2\mathbf{r} = 0
    \end{equation}
\end{enumerate}

Once $\mathbf{r}$ is obtained, one can get the solution of (\ref{Dc_T_optimization}) from $\mathbf{\widetilde{t}}^{\star}=\mathbf{M}^{-1}\mathbf{r}$.
Then the near optimal filter gain $\mathbf{T}(k)$ could be obtained by reshaping $\widetilde{\mathbf{t}}(k)$. The detailed steps for decentralized estimation are summarized in Algorithm.\ref{Algorithm.1}. 

The decentralized sequential estimation algorithm is detailed as follows.
\begin{enumerate}
    \item Each sensor estimates the observation matrices $\mathbf{H}_i(k)$  and the FC estimates the channel $\mathbf{G}(k)$ between the local sensors and the FC by using the pilot-based method same as for the centralized algorithm.
    \item At each time instant $k$, the FC updates the collaboration matrix and filter gain by solving (\ref{Solution of w}) and (\ref{Closed_form_T}). Then, the FC broadcasts updated $\mathbf{W}(k)$ and $\mathbf{T}(k)$ to all the local sensors.
    \item Local sensors which are tasked to communicate with the FC update their compression vectors locally by solving (\ref{decentralized_f}). 
    The post-collaboration observations are compressed using (\ref{Compressed_signal}) then the compressed observations are transmitted to the FC.
    \item Prior to transmitting the compressed data, local sensors first coherently transmit the packet headers which consist of $\mathbf{f}_i$ and $\mathbf{H}(k)$ to the FC.
\end{enumerate}

\subsection{Convergence analysis}
In this part, the convergence of the R-LMMSE estimator is analyzed. The following lemma shows the strict monotonicity of the proposed sequential estimator under certain condition.
\begin{lem}\label{lemma1}
As long as the designed  collaboration matrix $\mathbf{W}(k)$ and compression matrix $\mathbf{F}(k)$ satisfies $\mathbf{D}(k)\neq \mathbf{0}$,
MSE will strictly decrease with the update of $\mathbf{W}(k)$, $\mathbf{F}(k)$ and $\mathbf{T}(k)$. In other words, it will satisfy the following property
    \begin{equation}
        \Phi[\mathbf{W}(k-1),\mathbf{F}(k-1),\mathbf{T}(k-1) ] - \Phi[\mathbf{W}(k),\mathbf{F}(k),\mathbf{T}(k)] > 0
    \end{equation}
\end{lem}
\begin{proof}
    See Appendix \ref{Appendix: Prove of MSE convergence}
    \end{proof} 
As can be seen from Lemma \ref{lemma1}, the monotonicity of the R-LMMSE is not affected by the collaboration and compression strategies as long as the condition $\mathbf{D}(k)\neq \mathbf{0}$ is satisfied.
However, it is evident that the algorithm will converge faster with the suitably designed collaboration and compression strategies. In other words, by designing optimal or near-optimal collaboration and compression strategies, the rate of convergence of the estimator can be improved.
\subsection{Benchmark}
In this section, we present a benchmark algorithm to compare the performance of the proposed algorithms. 
Specifically, we assume that the FC has access to all the observations (uncompressed and without collaboration) from all $N$ sensors. Then, the R-LMMSE estimator of $\mathbf{x}$ for the benchmark system is given as
\begin{equation}
    \hat{\mathbf{x}}(k) = \hat{\mathbf{x}}(k-1) + \mathbf{T}(k)[\mathbf{y}(k) - \mathbf{H}(k)\hat{\mathbf{x}}(k-1)]
\end{equation}
where $\mathbf{y}(k) = \begin{bmatrix}\mathbf{y}_1^T(k) &\dots & \mathbf{y}_N^T(k)\end{bmatrix}^T\in \mathbb{R}^{NL\times 1}$. And the corresponding filter gain and error covariance update are given by
\begin{equation}
    \mathbf{T}(k) = \mathbf{P}(k-1)\mathbf{H}^T(k)[\mathbf{H}(k)\mathbf{P}(k-1)\mathbf{H}^T(k)+\mathbf{R}_{\mathbf{v}}]^{-1}
\end{equation}
\begin{equation}
    \mathbf{P}(k) = \mathbf{P}(k-1) - \mathbf{T}(k)\mathbf{H}(k)\mathbf{P}(k-1)
\end{equation}
Since for this system the FC makes use of all the observations from each sensor and provides the best achievable performance, it is reasonable to adopt this estimator as the benchmark.
\begin{algorithm}\label{Algorithm.1}
 {\bf Initialization:} $\mathbf{W}(0)$, $\mathbf{F}(0)$, $\mathbf{P}(0)$ \\
 {\bf While}{ $k>0$} \\
\indent {\bf ~~If}{ Compression matrix is computed {\em centrally}}
 {
 \begin{itemize}
\item Update collaboration matrix $\mathbf{W}(k)$ using (\ref{Solution of w})
 \item Update compression matrix $\{\mathbf{f}_i(k)\}_{i=1}^M$ at FC using (\ref{fi_central})
  \item Update filter gain matrix $\mathbf{T}(k)$ using (\ref{Closed_form_T})
  \item Update error-covariance matrix $\mathbf{P}(k)$ using (\ref{error_covariance})
 \end{itemize}
  }
\indent {\bf ~~If}{ Compression matrix is computed {\em locally}}
\begin{itemize}
    \item Update collaboration matrix $\mathbf{W}(k)$ using (\ref{Solution of w})
  \item Update compression matrix $\{\mathbf{f}_i(k)\}_{i=1}^M$ locally using (\ref{decentralized_f})
 \item Update filter gain matrix $\mathbf{T}(k)$ using (\ref{Dc_T_optimization})
  \item Update error covariance matrix $\mathbf{P}(k)$ using (\ref{error_covariance})
  \end{itemize}
  {\bf End}
\caption{Distributed Sequential MMSE estimation}
\end{algorithm}

\section{Numerical results}
\label{Sec: Numerical results}
In this section, we present several simulation results to demonstrate the effectiveness of our proposed algorithms. Specifically, the MSE performance as a function of various parameters is considered.

\begin{figure}[!t]
\centering
\includegraphics[width=3in, height=2.1 in]{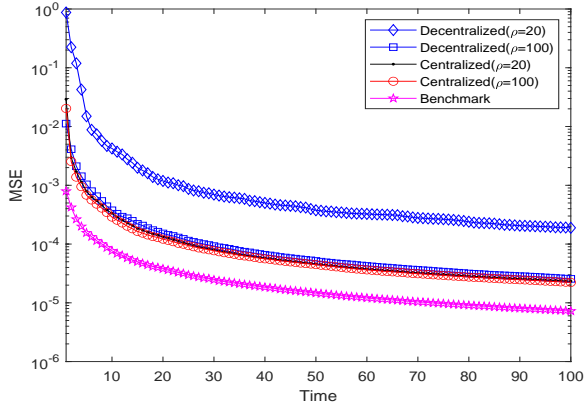}
\caption{MSE performance of different algorithms with respect to time where $\rho$ denotes the numbers of iteration during each time slot $k$.}
\label{fig_mse_time}
\end{figure}

\begin{figure}[!b]
\centering
\includegraphics[width=3in, height=2.1 in]{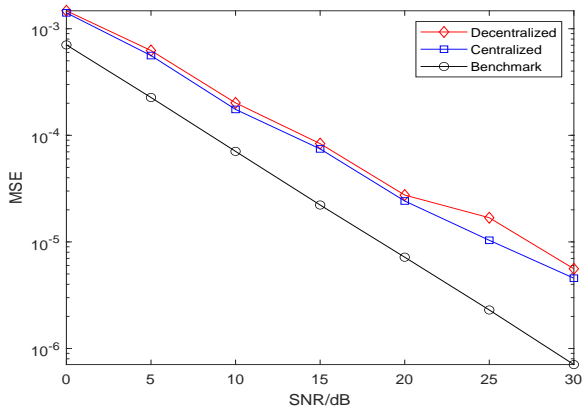}
\caption{MSE performance as a function of SNR in terms of observation noise.}
\label{fig_mse_snr1}
\end{figure}


\begin{figure}[!t]
\centering
\includegraphics[width=3in, height=2.1 in]{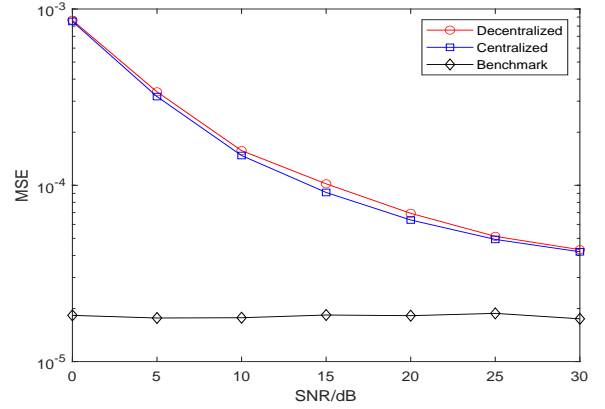}
\caption{MSE performance as a function of SNR in terms of communication noise at the FC.}
\label{fig_mse_snr3}
\end{figure}

\begin{figure}[!b]
\centering
\includegraphics[width=3in, height=2.1 in]{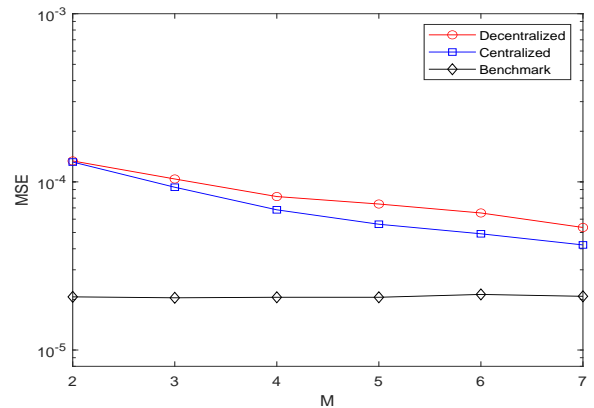}
\caption{MSE performance as a function in terms of number of sensors $M$ that are tasked to communicate with the FC.}
\label{fig_mse_M}
\end{figure}

For ease of comparison, a random vector $\mathbf{x}$ with zero mean and covariance matrix $\mathbf{R}_{\mathbf{x}}=\mathbf{I}_P$ is considered. At each time $k$, the individual elements of the observation matrix $\mathbf{H}(k)$ and channel matrix $\mathbf{G}(k)$ are assumed as zero mean and unit-variance Gaussian random variables. At the same time, the observation noise, $\mathbf{v}_i(k)$, collaboration noise, $\boldsymbol{\alpha}_i(k)$, and communication channel noise at the FC, $\boldsymbol{\epsilon}(k)$, are all assumed to be independent (spatially and temporally) zero mean Gaussian random vectors with covariance matrices $\mathbf{R}_{\mathbf{v}_i}= \sigma_{\mathbf{v}_i}^2\mathbf{I}_L$, $\mathbf{R}_{\boldsymbol{\alpha}_i} = \sigma_{\boldsymbol{\alpha}_i}^2
\mathbf{I}_L$ and $\mathbf{R}_{\boldsymbol{\epsilon}} = \sigma_{\boldsymbol{\epsilon}}^2
\mathbf{I}_S$, respectively. 
We define the SNR in terms of observation noise, collaboration noise and communication noise at the FC as  $1/\sigma_{\mathbf{v}_i}^2$,  $1/\sigma_{\boldsymbol{\alpha}_i}^2$ and $1/\sigma_{\boldsymbol{\epsilon}}^2$, respectively. In the following, the SNR across all of the channels is set as 20 dB, unless otherwise specified. 

Fig.\ref{fig_mse_time} presents the MSE performance of the proposed centralized algorithm and decentralized algorithm on a wireless sensor network with $P=3$, $L = 6$, $N = 7$ and $M = 3$. 
Also, the benchmark introduced in Section.III.F is used here for comparison. 
The sensor network topology is set as fully connected which means $\mathbf{A}_{mn}=1$ for all $m\in[1,M],n\in [1,N]$.   
As can be seen, both the centralized and decentralized algorithms perform well and the MSE  converges with time $k$. The MSE performance in the decentralized case is poorer compared with the centralized case as the solution of the filter gain in (\ref{Dc_T_optimization}) is a constrained optimization problem while in the centralized case it is an unconstrained problem. The number of iterations are set as $\rho =20$ and $\rho =100$ respectively for the two cases.
The interesting thing is that with more than 100 iterations, the decentralized algorithm can achieve 
almost the same performance as the centralized case which proves the effectiveness of the decentralized algorithm. 

\begin{figure}[!t]
\centering
\includegraphics[width=3in, height=2.1in]{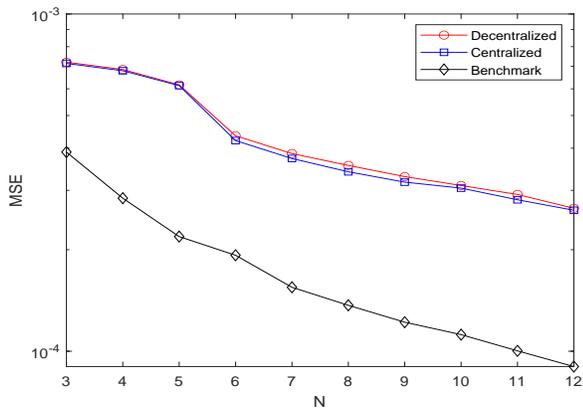}
\caption{MSE performance as a function of a number of sensors $N$ .}
\label{fig_mse_N}
\end{figure}

In Figures \ref{fig_mse_snr1} and \ref{fig_mse_snr3}, we plot the MSE performance of the proposed algorithms for different channel SNRs defined earlier in this section. Fig.\ref{fig_mse_snr1} shows the MSE as a function of measurement noise with $k=100$, $\rho =100$. Here, the time $k$ and the number of iterations $\rho$ are both kept sufficiently large to ensure that the algorithms have converged to sufficient accuracy. It can be seen that the MSE decreases with SNR as expected. MSE as a function of communication channel noise at the FC as shown in Fig.\ref{fig_mse_snr3} shows a similar behavior.

In Fig.\ref{fig_mse_M}, the MSE performance as a function of the number of sensors $M$ that could communicate with FC is given for $p = 3$, $N = 7$ for $k=30$, $\rho =10$. 
As can be seen, with the increase in the number of sensors that can communicate with FC, the MSE improves. This decrease in MSE is reasonable as the FC can access more information from the sensors. In Fig.\ref{fig_mse_N}, the MSE performance in terms of the number of sensors $N$ is plotted for $p=3$, $M=3$, $k=30$ and $\rho = 10$. 
As we can see, the MSE performance improves as $N$ increases. This behavior is also expected as now more sensors are collaborating in order to send information to the FC. 
In Fig.\ref{fig_mse_p}, the normalized MSE performance as a function of the signal dimension, $p$, is presented for $M=3$, $N=7$, $k=30$ and $\rho = 10$. The normalized MSE is defined as $\text{tr}(\mathbf{P}(k))/p$ for fair comparison.
It can be seen that with the increase in the parameter dimension, the MSE also increases as expected. This happens because the estimation problem becomes more and more difficult with the increase in signal dimension. 

\begin{figure}[!b]
\centering
\includegraphics[width=3in, height=2.1 in]{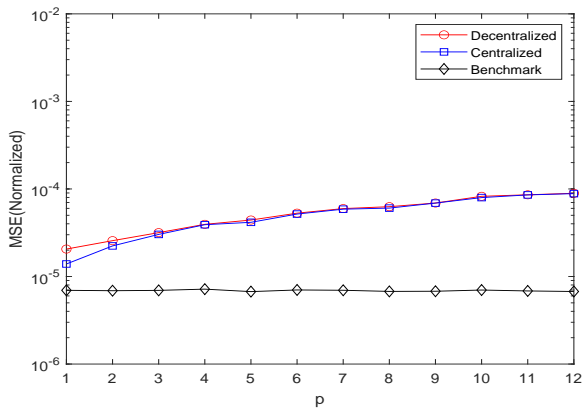}
\caption{Normalized MSE performance as a function of parameter dimension $p$.}
\label{fig_mse_p}
\end{figure}

\begin{figure}[!t]
\centering
\includegraphics[width=3.2in, height=2.2 in]{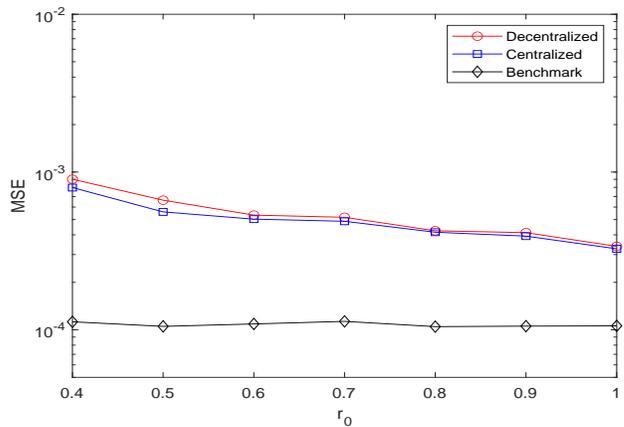}
\caption{MSE performance as the function of collaboration radius $r_0$.}
\label{fig_mse_w_topology}
\end{figure}

Finally, in Fig.\ref{fig_mse_w_topology}, a sparse sensor network topology is considered. The spatial placement and network structure is modeled as a random geometric graph, \cite{kar2013linear,freris2010fundamentals}, where sensors are assumed to be placed in a unit square meter area. All the sensors are only allowed to communicate with their neighbors when the distance between the neighbors is no more than $r_0$ meters. Here, $r_0$ denotes the collaboration radius. When 
the collaboration radius $r_0 = 1$, the network will be fully connected. As can be seen, the MSE performance improves with the increase of collaboration radius, $r_0$. 
Since more sensors are allowed to collaborate, thereby, more information will be transmitted to the FC for estimation.

\section{Estimation of Time-varying parameters}
\label{Sec: Time-varying parameters estimation}
In this section, distributed sequential estimation for tracking a time-varying parameter vector is considered. We assume that the state of the target follows
\begin{equation}
    \mathbf{x}(k) = \mathbf{A}_s(k-1)\mathbf{x}(k-1) + \mathbf{n}_s(k-1)
\end{equation}
where $\mathbf{A}_s(k-1)$ is the known state transition matrix and $\mathbf{n}_s(k-1)$ is the state noise with zero mean and covariance $\mathbf{R}_{\mathbf{n}_s}=\mathbb{E}[\mathbf{n}_s(k)\mathbf{n}_s^T(k)]$.

Similar to the derivation of sequential MMSE estimation of $\mathbf{x}$ in Sec.II.C, the prediction of $\mathbf{x}(k|k-1)$ and $\mathbf{P}(k|k-1)$ are given by
\begin{equation}
   \hat{\mathbf{x}}(k|k-1) = \mathbf{A}_s(k-1)\hat{\mathbf{x}}(k-1|k-1)
\end{equation}
\begin{equation}
    \mathbf{P}(k|k-1) =  \mathbf{A}_s(k-1)\hat{\mathbf{P}}(k-1|k-1)\mathbf{A}_s^T(k-1)+\mathbf{R}_{\mathbf{n}_s}
\end{equation}
The corresponding state update steps follow
\begin{align}
    &    \hat{\mathbf{x}}(k|k)  =\hat{\mathbf{x}}(k|k-1) +\mathbf{T}(k) \{\mathbf{q}(k)-\mathbb{E}[\mathbf{q}(k)|\Theta(k-1)]\} \\
       &    \mathbf{P}(k|k)  = [\mathbf{I} - \mathbf{T}(k)\mathbf{D}(k)]\mathbf{P}(k-1|k-1)[\mathbf{I} - \mathbf{T}(k)\mathbf{D}(k)]^T \nonumber\\
        & \qquad \qquad + \mathbf{T}(k)\mathbf{R}_n(k)\mathbf{T}^T(k)
\end{align}
where {$\mathbb{E}[\mathbf{q}(k)|\Theta(k-1)] = \mathbf{D}(k)\hat{\mathbf{x}}(k|k-1)$,} $\mathbf{D}(k) = \mathbf{G}(k)\mathbf{F}(k)[\mathbf{W}(k)\otimes \mathbf{I}_L]\mathbf{H}(k)$.
Then the estimation error covariance can expressed recursively as
\begin{equation}
\begin{split}
        \mathbf{P}(k|k)  = &\widetilde{\mathbf{D}}(k)\mathbf{A}_s(k-1)\hat{\mathbf{P}}(k-1|k-1)\mathbf{A}_s^T(k-1)\widetilde{\mathbf{D}}^T(k)\\
        &+ \widetilde{\mathbf{D}}(k)\mathbf{R}_{\mathbf{n}_s}\widetilde{\mathbf{D}}^T(k) +\mathbf{T}(k)\mathbf{R}_n(k)\mathbf{T}^T(k)
\end{split}
\end{equation}
where $\widetilde{\mathbf{D}}(k) = \mathbf{I} - \mathbf{T}(k)\mathbf{D}(k)$.

We can also apply the proposed algorithm for this case. The method in Sec.III could be used to obtain the optimal collaboration and compression strategies.
Thus, one can obtain the solution for the time-varying parameters estimation problem following a similar process.

\section{Conclusion}
\label{Sec: Conclusion}
This paper focused on the problem of distributed sequential estimation of a random parameter vector in a resource constrained WSN. A communication efficient collaboration-compression framework was proposed for solving this problem. Specifically, the local sensors first collaborate (via a collaboration matrix) with each other and then a subset of sensors transmit the observations obtained after collaboration to the FC. Importantly, before transmission to the FC the observations at the local sensors are compressed to reduce the communication costs further. Near-optimal collaboration and linear compression strategies are designed jointly for the goal of recursively minimizing the mean square error. 
Further, we show that even though the work focused on estimating random vectors, the proposed methods can be used for estimating time-varying random vector parameters with a known transition matrix. 

Future extensions of this work include, power allocation for individual sensors while designing efficient collaboration-compression strategies in the WSNs. Also, the problem of optimal topology design under the proposed framework is an interesting research direction. Moreover, quantization based schemes for collaboration and compression will also be an interesting future research direction. 
\appendices
\section{Constraint function}
\label{Appendix: Constraint function}
 Let $\textbf{w}\in \mathbb{R}^U$ be the nonzero elements of $\text{vec}(\textbf{W})$ where $\textbf{W}\in \mathbb{R}^{M\times N}$. Apparently, each element in $\textbf{w}$ uniquely correspond to $\textbf{W}$, which can be noted as $w_u = W_{m_u n_u}$ where $m_u\in [1,M]$, $n_u \in [1,N]$, $u\in [1,U]$.
 
 Given $\textbf{a}\in \mathbb{R}^{ML}$, we can get
 \begin{equation}
     \textbf{a}^T (\textbf{W}\otimes \textbf{I}_L) = [\textbf{a}^T[\textbf{W}\otimes \textbf{I}_L]_{.1}, \dots, \textbf{a}^T[\textbf{W}\otimes \textbf{I}_L]_{.NL} ]
 \end{equation}
where $\textbf{a}^T[\textbf{W}\otimes \textbf{I}_L]_{.j}$ represents the $j$th column of $\textbf{a}^T[\textbf{W}\otimes \textbf{I}_L] $. 

Meanwhile, given $\textbf{A}\in \mathbb{R}^{U\times NL}$, 
\begin{equation}
    \textbf{w}^T \textbf{A} = [\textbf{w}^T \textbf{A}_{.1},\dots,\textbf{w}^T \textbf{A}_{.NL}]
\end{equation}
where $\textbf{A}_{.j}$ is the $j$th column of $\textbf{A}$ and $\textbf{w}^T \textbf{A}_{.j} = \sum_{u=1}^{U}{W_{m_u n_u}A_{uj}}$ for $j\in [1,NL]$. 

Consider the $j$th entry of $\textbf{w}^T \textbf{A}$, we can obtain
\begin{equation}
    \begin{split}
        [\textbf{w}^T \textbf{A}]_{j} &= \sum_{u=1}^{U}{W_{m_u n_u}A_{uj}}\\
        &= \sum_{u=1,n_u=\lfloor{\frac{j-1}{L}}\rfloor +1 }^U {a_{L(m_u - 1)+ j-L\lfloor{\frac{j-1}{L}}\rfloor)} W_{m_u j}}\\
        &= \sum_{m_u = 1}^M {a_{L(m_u - 1)+ j-L\lfloor{\frac{j-1}{L}}\rfloor }W_{m_u j} } \\
        &= [\textbf{a}^T(\textbf{W}\otimes \textbf{I}_L)]_{j}
    \end{split}
\end{equation}
where we have made use of the fact that 
\begin{equation}
    A_{uj}=\left\{\begin{matrix}a_{L(m_u - 1)+ j-L\lfloor{\frac{j-1}{L}}\rfloor) },& n_u=\lfloor{\frac{j-1}{L}}\rfloor +1 \\ 
    0,& \text{otherwise} 
\end{matrix}\right.
\end{equation}

Next, we will show the provement of property (\ref{map_property2}). Given $\textbf{B}\in \mathbb{R}^{P\times ML}$ , $\textbf{C}\in \mathbb{R}^{NL\times NL}$ and $\textbf{D}\in \mathbb{R}^{ML\times P}$, we can obtain
that 
\begin{equation}
    \text{tr}[\textbf{B}(\textbf{W}\otimes \textbf{I}_L)\textbf{C}(\textbf{W}\otimes \textbf{I}_L)^T\textbf{D}] = \sum_{i=1}^P{\textbf{e}_i^T}\textbf{B}(\textbf{W}\otimes \textbf{I}_L)\textbf{C}(\textbf{W}\otimes \textbf{I}_L)^T\textbf{D}\textbf{e}_i
\end{equation}
where $\mathbf{B} = \begin{bmatrix}
\mathbf{b}_1 & \cdots &\mathbf{b}_P
\end{bmatrix}^T$ , $\mathbf{D} = \begin{bmatrix}\mathbf{d}_1&\cdots&\mathbf{d}_P\end{bmatrix}$ and $\mathbf{e}_i$ is the base vector whose entries are zero except that the $i$th entry is 1.

By using the property (\ref{map_property1}), we have
\begin{equation}
    \textbf{e}_i^T\textbf{B}(\textbf{W}\otimes \textbf{I}_L) = \textbf{b}_i^T(\textbf{W}\otimes \textbf{I}_L) = \textbf{w}^T\widetilde{\textbf{B}}_i
\end{equation}
\begin{equation}
    (\textbf{W}\otimes \textbf{I}_L)^T\textbf{D}\textbf{e}_i = (\textbf{W}\otimes \textbf{I}_L)^T\textbf{d}_i = \widetilde{\mathbf{D}}_i^T\textbf{w}
\end{equation}
where 
\begin{equation}
        [\widetilde{\textbf{B}}_i]_{uj}=\left\{\begin{matrix}
 [\textbf{b}_i]_{L(m_u - 1)+ j-L\lfloor{\frac{j-1}{L}}\rfloor },& n_u=\lfloor{\frac{j-1}{L}}\rfloor +1 \\ 
 0,& \text{otherwise} 
\end{matrix}\right.
\end{equation}
and
\begin{equation}
        [\widetilde{\textbf{D}}_i]_{uj}=\left\{\begin{matrix}
 [\textbf{d}_i]_{L(m_u - 1)+ j-L\lfloor{\frac{j-1}{L}}\rfloor },& n_u=\lfloor{\frac{j-1}{L}}\rfloor +1 \\ 
 0,& \text{otherwise} 
\end{matrix}\right.
\end{equation}
for $u\in [1,U]$ and $j\in [1,NL]$

Therefore, 
\begin{equation}
    \text{tr}[\textbf{B}(\mathbf{W}\otimes \mathbf{I}_L)\mathbf{C}(\mathbf{W}\otimes \mathbf{I}_L)^T\mathbf{D}] = \sum_{i=1}^{r_B} \mathbf{w}^T\widetilde{\mathbf{B}}_i \mathbf{C} \widetilde{\mathbf{D}}_i^T\mathbf{w} = \mathbf{w}^T \mathbf{E}\mathbf{w}
\end{equation}
where $\mathbf{E} = \sum_{i=1}\widetilde{\mathbf{B}}_i \mathbf{C} \widetilde{\mathbf{D}}_i^T$, $r_B$ represents the number of rows of matrix $\mathbf{B}$, then we can obtain property (\ref{map_property2}). Similarly, 
\begin{equation}
     \text{tr}[\mathbf{B}(\mathbf{W}\otimes \mathbf{I})\textbf{C}] = \sum_{i=1}\textbf{e}_i^T\textbf{B}(\textbf{W}\otimes \textbf{I})\textbf{C}\textbf{e}_i = \textbf{w}^T(\sum_{i=1}\widetilde{\mathbf{B}}_i \textbf{c}_i) = \textbf{w}^T \widetilde{\mathbf{c}}_i
\end{equation}
where $\textbf{c}_i$ is the $i$th column of $\textbf{C}$.

\section{Coefficient matrix}
\label{Appendix: Coefficient matrix}
Based on proposition 1, the problem in (\ref{Target_function}) can be expressed as the quadratic function of $\textbf{w}$. Recall the expression of error corvariance is given by 
\begin{equation}\label{error_covariance_expression}
    \begin{split}
       & \text{tr}[\mathbf{P}(k)]\\  =&\text{tr}\{ [\mathbf{I} - \mathbf{T}(k)\mathbf{D}(k)]\mathbf{P}(k-1)[\mathbf{I} - \mathbf{T}(k)\mathbf{D}(k)]^T \\
        &+ \mathbf{T}(k)\mathbf{R}_n(k)\mathbf{T}(k)^T\}\\
        =& \text{tr}[\mathbf{P}(k-1)]+\mathbf{tr}[\mathbf{T}(k)\mathbf{R}_{\boldsymbol{\epsilon}} \mathbf{T}^T(k)]\\
        &+ \text{tr}\{\mathbf{T}(k)\mathbf{G}\mathbf{F}(k)[\mathbf{W}(k)\otimes \mathbf{I}_L]\mathbf{H}(k)\mathbf{P}(k-1)\mathbf{H}^T(k)\\&\qquad [\mathbf{W}(k)\otimes\mathbf{I}_L]^T\mathbf{F}^T(k)\mathbf{G}^T\mathbf{T}^T(k)\}\\
        &-\text{tr}\{\mathbf{T}(k)\mathbf{G}\mathbf{F}(k)[\mathbf{W}(k)\otimes \mathbf{I}_L]\mathbf{H}(k)\mathbf{P}(k-1) \}\\
        &-\text{tr}\{\mathbf{P}(k-1)\mathbf{H}^T(k)[\mathbf{W}(k)\otimes\mathbf{I}_L]^T\mathbf{F}^T(k)\mathbf{G}^T\mathbf{T}^T(k) \}\\
        &+ \text{tr}\{\mathbf{T}(k)\mathbf{G}\mathbf{F}(k)[\mathbf{W}(k)\otimes\mathbf{I}_L]\mathbf{R}_{\mathbf{v}}[\mathbf{W}(k)\otimes\mathbf{I}_L]^T\mathbf{F}^T(k)\\&\qquad \mathbf{G}^T\mathbf{T}^T(k)\}\\
        &+ \text{tr}[\mathbf{T}(k)\mathbf{G}\mathbf{F}(k) \mathbf{R}_{\alpha}\mathbf{F}^T(k)\mathbf{G}^T\mathbf{T}^T(k)]\\
        \end{split} 
\end{equation}

Let $\mathbf{B}_0^{(0)}=\mathbf{T}(k)\mathbf{G}\mathbf{F}(k) $, $\mathbf{C}_0^{(0)}=\mathbf{H}(k)\mathbf{P}(k-1)\mathbf{H}^T(k)+\mathbf{R}_{\mathbf{v}}$, $\mathbf{D}_0^{(0)}=\mathbf{F}^T(k)\mathbf{G}^T\mathbf{T}^T(k)$ and $\mathbf{C}_1^{(0)}=\mathbf{H}(k)\mathbf{P}(k-1)$, according to (\ref{map_property2}) and (\ref{map_property3}), we can get
\begin{equation}\label{expression_Omega_0}
        \text{tr}[\mathbf{B}_0^{(0)}(\mathbf{W}\otimes \mathbf{I}_L)\mathbf{C}_0^{(0)}(\textbf{W}\otimes \mathbf{I}_L)^T\mathbf{D}_0^{(0)}] = \mathbf{w}^T \boldsymbol{\Omega}^{(0)} \mathbf{w}
\end{equation}
\begin{equation}
    \text{tr}[\mathbf{B}_0^{(0)}(\mathbf{W}\otimes \mathbf{I}_L)\mathbf{C}_1^{(0)}] = \mathbf{w}^T {\mathbf{d}}
\end{equation}
and the constant term $\eta_0$ is given by 
\begin{equation}
    \begin{split}
       \eta_0  =& \text{tr}[\mathbf{P}(k-1)]+\mathbf{tr}[\mathbf{T}(k)\mathbf{R}_{\boldsymbol{\epsilon}} \mathbf{T}^T(k)]\\&+ \text{tr}[\mathbf{T}(k)\mathbf{G}\mathbf{F}(k) \mathbf{R}_{\alpha}\mathbf{F}^T(k)\mathbf{G}^T\mathbf{T}^T(k)]\\
        \end{split} 
\end{equation}
Then, the target function in terms of $\mathbf{W}(k)$ can be represented as the function of $\mathbf{w}(k)$ as follows
\begin{equation}
    \begin{split}
        \text{tr}[\textbf{P}(k)] 
        =& \textbf{w}^T\boldsymbol{\Omega}^{(0)}\textbf{w} - 2\textbf{w}^T\textbf{d}+ \eta_0
        \end{split} 
\end{equation}

At the same time, the expected energy cost for sensor collaboration in (\ref{collaboration_constraint})
is given by
\begin{equation}\label{collaboration_constraint_appendix}
        \mathbb{E}[{C}_i^{(1)}(k)]  = \text{tr}(\mathbf{R}_{\mathbf{y}_i})\{\mathbf{e}_i^T[\mathbf{W}(k)\odot \widetilde{\mathbf{I}}]^T[\mathbf{W}(k)\odot \widetilde{\mathbf{I}}]\mathbf{e}_i\}
\end{equation}
Let $\widetilde{\mathbf{W}}(k) = \mathbf{W}(k)\odot \widetilde{\mathbf{I}}$, where $\widetilde{\mathbf{W}}(k)$ is same as $\mathbf{W}(k)$ except that the diagonal elements are set as 0. Then (\ref{collaboration_constraint_appendix})  can be reexpressed as:
\begin{equation}\label{E1_constraint_w}
    \begin{split}
         \mathbb{E}[{C}_i^{(1)}(k)]  =&  \text{tr}(\mathbf{R}_{\mathbf{y}_i})\text{tr}\{[\mathbf{W}(k)\odot \widetilde{\mathbf{I}}]\mathbf{e}_i\mathbf{e}_i^T[\mathbf{W}(k)\odot \widetilde{\mathbf{I}}]^T\}\\
       = &\text{tr}(\mathbf{R}_{\mathbf{y}_i}) \text{tr} [\widetilde{\mathbf{W}}(k) \mathbf{e}_i \mathbf{e}_i^T \widetilde{\textbf{W}}^T(k)]\\
    \overset{(a)}{=}  & \text{tr}(\mathbf{R}_{\mathbf{y}_i})\widetilde{\mathbf{w}}^T \mathbf{E}_i \widetilde{\mathbf{w}}
      \overset{(b)}{=} \text{tr}(\mathbf{R}_{\mathbf{y}_i})\mathbf{w}^T \mathbf{J}^T \mathbf{E}_i \mathbf{J} \mathbf{w}
    \end{split}
\end{equation}
where the transition of (a) can be obtained by 
\begin{equation}
     \text{tr}[\mathbf{B}_0^{(1)}(\widetilde{\mathbf{W}}\otimes \mathbf{I}_L)\mathbf{C}_0^{(1)}(\widetilde{\mathbf{W}}\otimes \mathbf{I}_L)^T\mathbf{D}_0^{(1)}] = \widetilde{\mathbf{w}}^T \boldsymbol{\Omega}^{(1)} \widetilde{\mathbf{w}}
\end{equation}
where $\mathbf{B}_0^{(1)} =\mathbf{D}_0^{(1)}= \mathbf{I}_M$, $\mathbf{C}_0^{(1)}=\mathbf{e}_i \mathbf{e}_i^T$,$L=1$ and
$s$ is the number of nonzero entries in $\widetilde{\mathbf{W}}(k)$, and $\widetilde{\mathbf{w}}\in \mathbb{R}^{s\times 1}$ is the vector consists of the nonzero element of $\widetilde{\mathbf{W}}(k)$. At the same time, $\mathbf{w}$ consists of the nonzero elements of $\mathbf{W}(k)$, which means $\widetilde{\mathbf{w}}$ can be linearly transformed into $\mathbf{w}$ as (b). And $\mathbf{J}$ is given by  
\begin{equation}
    [\textbf{J}]_{ij} = \left\{\begin{matrix}
     1,& \mathcal{L}(i)=j\\
     0,& \text{otherwise}
    \end{matrix}\right.
\end{equation}
where $\mathcal{L}=\{u | w_u = W_{m_u n_u},m_u \neq n_u\}$.

In order to make it clear, a simple example is provided as follows.
Assume the collaborative matrix $\mathbf{W}(k)$ and $\widetilde{\mathbf{W}}(k)$ are given by
\begin{equation}
    \mathbf{W}=\begin{bmatrix}
 w_1&  0& 0 &  w_5&  w_7&  0\\ 
 w_2&  w_3&  0&  0&  w_8&  w_9\\ 
 0&  0&  w_4&  w_6&  0&  0
\end{bmatrix}
\end{equation}
\begin{equation}
\widetilde{\mathbf{W}}=\begin{bmatrix}
 0&  0& 0 &  w_5&  w_7&  0  \\ 
 w_2&  0&  0&  0&  w_8&  w_9 \\ 
 0&  0&  0&  w_6&  0&  0
\end{bmatrix}
\end{equation}
then, the corresponding vectors are given by
\begin{equation}
    \mathbf{w}=\begin{bmatrix}
 w_1&  w_2& w_3& w_4& w_5& w_6& w_7&  w_8& w_9\
\end{bmatrix}^T
\end{equation}
\begin{equation}
    \widetilde{\mathbf{w}}=\begin{bmatrix}
  w_2& w_5& w_6& w_7&  w_8& w_9\
\end{bmatrix}^T
\end{equation}
and the corresponding transition matrix $\mathbf{J}$ is given by
\begin{equation}
    \mathbf{J}=\begin{bmatrix}
0 &  1&  0&  0&  0&  0&  0&  0& 0\\ 
0 &  0&  0&  0&  1&  0&  0&  0& 0\\ 
0 &  0&  0&  0&  0&  1&  0&  0& 0\\ 
0 &  0&  0&  0&  0&  0&  1&  0& 0\\ 
0 &  0&  0&  0&  0&  0&  0&  1& 0\\ 
0 &  0&  0& 0 &  0&  0&  0&  0& 1
\end{bmatrix}
\end{equation}

Thus, the coefficient matrix $\boldsymbol{\Omega}^{(1)}$ is given by 
\begin{equation}
    \boldsymbol{\Omega}^{(1)} = \text{tr}(\mathbf{R}_{\mathbf{y}_i})\mathbf{J}^T \mathbf{E}_i \mathbf{J} 
\end{equation}

The compression cost is given by 
\begin{equation}
    \begin{split}
        &\mathbb{E}[{C}_i^{(2)}(k)] \\
        =& \mathbf{f}_i^T(\mathbf{e}_i^T\otimes \mathbf{I}_L)[\mathbf{W}(k)\otimes \mathbf{I}_L] \mathbf{R}_{Y}[\mathbf{W}(k)\otimes \mathbf{I}_L]^T (\mathbf{e}_i\otimes \mathbf{I}_L)\mathbf{f}_i\\
        &+ \mathbf{f}_i^T(\mathbf{e}_i^T\otimes \mathbf{I}_L)\mathbf{R}_{\boldsymbol{\alpha}}(\mathbf{e}_i\otimes \mathbf{I}_L)\mathbf{f}_i\\
        =& \text{tr}\{(\mathbf{e}_i\otimes \mathbf{I}_L)\mathbf{f}_i \mathbf{f}_i^T(\mathbf{e}_i^T\otimes \mathbf{I}_L)[\mathbf{W}(k)\otimes \mathbf{I}_L] \mathbf{R}_{Y}[\mathbf{W}(k)\otimes \mathbf{I}_L]^T\}\\
        &+ \mathbf{f}_i^T(\mathbf{e}_i^T\otimes \mathbf{I}_L)\mathbf{R}_{\boldsymbol{\alpha}}(\mathbf{e}_i\otimes \mathbf{I}_L)\mathbf{f}_i\\
    \end{split}
\end{equation}

Let $\mathbf{B}_0^{(2)} = (\mathbf{e}_i\otimes \mathbf{I}_L)\mathbf{f}_i \mathbf{f}_i^T(\mathbf{e}_i^T\otimes \mathbf{I}_L)$, $\mathbf{C}_0^{(2)} =\mathbf{R}_\mathbf{y} $, and $\mathbf{D}_0^{(2)} = \mathbf{I}_{ML}$, then $\boldsymbol{\Omega^{(2)}}$ and $\eta_1$ can be given by
\begin{equation}\label{E2_constraint_w}
        \text{tr}[\mathbf{B}_0^{(2)}(\mathbf{W}\otimes \mathbf{I}_L)\mathbf{C}_0^{(2)}(\textbf{W}\otimes \mathbf{I}_L)^T\mathbf{D}_0^{(2)}] = \mathbf{w}^T \boldsymbol{\Omega}^{(2)} \mathbf{w}
\end{equation}
\begin{equation}
    \eta_1 = \mathbf{f}_i^T(\mathbf{e}_i^T\otimes \mathbf{I}_L)\mathbf{R}_{\boldsymbol{\alpha}}(\mathbf{e}_i\otimes \mathbf{I}_L)\mathbf{f}_i
\end{equation}
Then, all the coefficient is problem (\ref{Solution of w}) are provided.

Consider the positive definiteness of the coefficient matrix. Based on (\ref{expression_Omega_0}) and (\ref{expression_E}), for any nonzero vector $\mathbf{r}$
\begin{equation}\label{Prove_positive_omega0}
\begin{split}
    &\mathbf{r}^T\boldsymbol{\Omega}^{(0)}\mathbf{r} = \mathbf{r}^T \sum_{i=1}^{r_B}[\widetilde{\mathbf{B}}_0^{(0)}]_i \mathbf{C}_0^{(0)} [\widetilde{\mathbf{B}}_0^{(0)}]_i^T  \mathbf{r}\\
    =& \mathbf{r}^T \sum_{i=1}^{r_B}[\widetilde{\mathbf{B}}_0^{(0)}]_i [\mathbf{H}(k)\mathbf{P}(k-1)\mathbf{H}^T(k)+\mathbf{R}_{\mathbf{v}}] [\widetilde{\mathbf{B}}_0^{(0)}]_i^T  \mathbf{r}\\ 
    =&\mathbb{E}\{\mathbf{r}^T \sum_{i=1}^{r_B}[\widetilde{\mathbf{B}}_0^{(0)}]_i \mathbf{H}(k)\mathbf{e}(k-1)\mathbf{e}^T(k-1)\mathbf{H}^T(k) [\widetilde{\mathbf{B}}_0^{(0)}]_i^T  \mathbf{r}\}\\
    &+\mathbb{E}\{\mathbf{r}^T \sum_{i=1}^{r_B}[\widetilde{\mathbf{B}}_0^{(0)}]_i \mathbf{v}(k)\mathbf{v}^T(k) [\widetilde{\mathbf{B}}_0^{(0)}]_i^T  \mathbf{r}\}\\
    =&\sum_{i=1}^{r_B}\mathbb{E}\big\{ \{\mathbf{r}^T [\widetilde{\mathbf{B}}_0^{(0)}]_i \mathbf{H}(k)\mathbf{e}(k-1)\}^2\big \}\\
    &+\sum_{i=1}^{r_B}\mathbb{E}\big \{ \{\mathbf{r}^T [\widetilde{\mathbf{B}}_0^{(0)}]_i \mathbf{v}(k)\}^2\big \}>  0\\
\end{split}
\end{equation}
which means that $\boldsymbol{\Omega}^{(0)}$ is a positive definite matrix.

Then, consider $\boldsymbol{\Omega}^{(1)}$ and $\boldsymbol{\Omega}^{(2)}$. Observe (\ref{E1_constraint_w}) and (\ref{E2_constraint_w}), for any nonzero vector $\mathbf{w}$, there always exists a corresponding $\mathbf{W}(k)$ that can make the equalities hold. More specifically, given network topology $\mathbf{A}$, any vector $\mathbf{w}$ can be mapped into $\mathbf{W}(k)$ through (\ref{map_w_W}).
Meanwhile, from (\ref{Collaboration_cost}) and (\ref{compression_cost}), it is evident that 
\begin{equation}
    \mathbb{E}[{C}_i^{(1)}(k)] \geq 0,\quad \mathbb{E}[{C}_i^{(2)}(k)]\geq 0
\end{equation}
always hold with equality when is no communication.

That is to say, for any nonzero vector $\mathbf{w}$, 
\begin{equation}
    \mathbf{w}^T\boldsymbol{\Omega}_i^{(1)}\mathbf{w} >0 ~~
\text{and}~~
    \mathbf{w}^T\boldsymbol{\Omega}_i^{(2)}\mathbf{w} >0 
\end{equation}
always hold, which implies that $\boldsymbol{\Omega}^{(1)}$ and $\boldsymbol{\Omega}^{(2)}$ are also positive definite matrices.

Consider the proof of positive semi-definiteness of $\boldsymbol{\Omega}^{(3)}$. For any nonzero vector $\mathbf{r}$,
\begin{equation}
\begin{split}
      \mathbf{r}^T\boldsymbol{\Omega}^{(3)}\mathbf{r}=&\pi_{ii}\mathbf{r}^T\big[\mathbf{W}_i(k)\mathbf{H}(k)\mathbf{P}(k-1)\mathbf{H}^T(k)\mathbf{W}_i^T(k)\\
      & + \mathbf{W}_i(k)\mathbf{R}_{\mathbf{v}}\mathbf{W}_i^T(k)+\mathbf{R}_{\alpha_i}\big]\mathbf{r} 
\end{split}
\end{equation}
where $ \pi_{ii} =\|\mathbf{T}(k)\mathbf{g}_i(k)\|_2^2 \geq 0
$, and 
\begin{equation}
\begin{split}
      &\mathbf{r}^T\big[\mathbf{W}_i(k)\mathbf{H}(k)\mathbf{P}(k-1)\mathbf{H}^T(k)\mathbf{W}_i^T(k) \\&+ \mathbf{W}_i(k)\mathbf{R}_{\mathbf{v}}\mathbf{W}_i^T(k)+\mathbf{R}_{\alpha_i}\big]\mathbf{r}  \\
      =&\mathbb{E}\{[\mathbf{r}^T\mathbf{W}_i(k)\mathbf{H}(k)\mathbf{e}(k-1)]^2\}+ \mathbb{E}\{[\mathbf{r}^T\mathbf{W}_i(k)\mathbf{v}(k)]^2\} \\&+ \mathbb{E}\{[\mathbf{r}^T\mathbf{\alpha_i}(k)]^2\}  > 0
\end{split}
\end{equation}
therefore, $\mathbf{r}^T\boldsymbol{\Omega}^{(3)}\mathbf{r}\geq 0$ holds for any nonzero vector $\mathbf{r}$. Therefore, $\boldsymbol{\Omega}^{(3)}$ is a positive semi-definite matrix.
\section{Prove of MSE convergence}
\label{Appendix: Prove of MSE convergence}
Before we show the convergence of the proposed algorithm, a lemma will be used in our provement is provided \cite{baksalary1992inequality}. 
\begin{lem}
    For any matrix $\mathbf{A}$  and $\mathbf{B}$ which are symmetric and non-negative definite, have the following property
    \begin{equation}
        \lambda_{\text {min}}(\textbf{A}) \text{tr}(\textbf{B}) \leq \text{tr}(\textbf{AB}) \leq \lambda_{\text{max}}(\mathbf{A})\text{tr}(\textbf{B})
    \end{equation}
    where $\lambda_{\text{min}}$ and $\lambda_{\text{max}}$ are the smallest and biggest eigenvalue of matrix $\textbf{A}$ respectively.
\end{lem}

In the centralized case, the decoder can be expressed in closed form as follows.
\begin{equation}
\begin{split}
    \mathbf{T}(k)=& \mathbf{P}(k-1)\mathbf{H}^T[\mathbf{W}(k)\otimes \mathbf{I}]^T \mathbf{F}^T(k) \mathbf{G}^T \\ &[\mathbf{D}(k)  \mathbf{P}(k-1)
      \mathbf{D}^T(k) 
      +\mathbf{R}_n(k) ]^{-1}\\
\end{split}
\end{equation}

Denote $\textbf{R}_a = [\mathbf{D}(k)  \mathbf{P}(k-1)\mathbf{D}^T(k)      +\mathbf{R}_n(k) ]^{-1}$, then $\textbf{T}(k)$ can be compactly expressed as
\begin{equation}
    \mathbf{T}= \mathbf{P}(k-1)\mathbf{H}^T[\mathbf{W}(k)\otimes \mathbf{I}]^T \mathbf{F}^T \mathbf{G}^T \mathbf{R}_a
\end{equation}
     
At the same time, 
\begin{equation}
    \mathbf{P}(k) = \mathbf{P}(k-1) - \mathbf{T}(k) \mathbf{G} \mathbf{F}(k) [\mathbf{W}(k)\otimes \mathbf{I}] \mathbf{H} \mathbf{P}(k-1)
\end{equation}

then 
\begin{equation}\label{diff_of_mse}
\begin{split}
        &  \Phi(k-1) - \Phi(k) \\
       =&  \text{tr}\{\mathbf{P}(k-1)\mathbf{H}^T[\mathbf{W}(k)\otimes \mathbf{I}]^T \mathbf{F}^T \mathbf{G}^T \mathbf{R}_a \mathbf{G} \mathbf{F} [\mathbf{W}\otimes \mathbf{I}] \mathbf{H} \mathbf{P}(k-1) \}  \\
       \geq & \lambda_{\text{min}}[\textbf{P}^2(k-1)] \text{tr}\{\mathbf{H}^T[\textbf{W}(k)\otimes \mathbf{I}]^T \mathbf{F}^T \mathbf{G}^T \mathbf{R}_a \mathbf{G} \mathbf{F} [\mathbf{W}\otimes \mathbf{I}] \mathbf{H}\}\\
       \geq & \lambda_{\text{min}}^2[\textbf{P}(k-1)] \lambda_{\text{min}}(\mathbf{R}_a) \text{tr}[\mathbf{D}(k) \textbf{D}^T(k)]\\
       = & \lambda_{\text{min}}^2[\mathbf{P}(k-1)] \lambda_{\text{min}}(\mathbf{R}_a) \|\mathbf{D} \|_F^2
\end{split}
\end{equation}

From (\ref{diff_of_mse}), it can be seen that MSE will strictly decrease with the update of $\mathbf{W}(k)$, $\mathbf{F}(k)$ and $\mathbf{T}(k)$ as long as the designed collaboration matrix $\mathbf{W}(k)$ and compression matrix $\mathbf{F}(k)$ could satisfy $\mathbf{D}(k)\neq \mathbf{0}$ and   $\lambda_{\text{min}}(\textbf{R}_a) > 0$. 

Notice that $\textbf{R}_a^{-1}$ is symmetric and positive definite, which means that 
\begin{equation}
    \lambda(\textbf{R}_a) = \frac{1}{ \lambda(\textbf{R}_a^{-1})} >0
\end{equation}

Therefore,  we can obtain 
\begin{equation}
    \Phi(k-1) - \Phi(k) > 0.
\end{equation}


\ifCLASSOPTIONcaptionsoff
  \newpage
\fi



%

\bibliographystyle{IEEEtran}
\bibliography{IEEEabrv,references}
\end{document}